\documentclass[11pt]{article}
\pdfoutput=1

\usepackage[dvipsnames]{xcolor}
\usepackage{hyperref}  
\hypersetup{
     colorlinks,
     citecolor=OliveGreen,
     filecolor=black,
     linkcolor=blue,
     urlcolor=black,
     linktoc=all
}
\usepackage[shortlabels]{enumitem} 
\usepackage{amsmath, amssymb, amsthm, mathtools}

\newcommand{\ignore}[1]{}
\newtheorem{theorem}{Theorem}
\newtheorem{lemma}{Lemma}[section]

\newtheorem{observation}{Observation}

\newtheorem{definition}{Definition}
\newtheorem{remark}{Remark}

\usepackage{bbm}
\usepackage{mathrsfs}
\usepackage{latexsym}
\usepackage{microtype}
\usepackage{xspace}
\usepackage{framed}
\usepackage{tocloft}
\PassOptionsToPackage{dvipsnames,usenames,table}{xcolor}
\usepackage{tikz}
\usetikzlibrary{calc,math}
\usetikzlibrary{arrows,shapes}
\usepackage{lineno} 

\newbox\mytempbox
\tikzstyle{embeds} = [->, >=open triangle 45]

\usepackage{xpunctuate}
\usepackage[all,british]{foreign} 
\redefnotforeign[ie]{i.e\xperiodafter} 
\redefnotforeign[eg]{e.g\xperiodafter}

\newenvironment{tightcenter}
 {\parskip=0pt\par\nopagebreak\centering}
 {\par\noindent\ignorespacesafterend}

\usepackage{marginnote} 

\usepackage[draft,english,silent]{fixme} 
\fxsetup{theme=color,layout=marginnote,innerlayout=noinline}

\usepackage{amsmath, amsthm, amssymb, amsfonts} 
\usepackage{thmtools} 
\usepackage{mathtools}

\usepackage[full,small]{complexity}

\usepackage{environ}

\usepackage{ifthen}
\newboolean{NotfullVersion}

\setboolean{NotfullVersion}{false} 

\newcommand{\Naturals}{\mathbb{N}}

\newcommand{\eps}{\epsilon}

\def\adm_#1{ \operatorname{adm}_{#1} }
\def\wcol_#1{ \operatorname{wcol}_{#1} }
\def\scol_#1{ \operatorname{scol}_{#1} }

\newtheorem*{theorem*}{Theorem}
\newlength{\RoundedBoxWidth}
\newsavebox{\GrayRoundedBox}
   {\setlength{\RoundedBoxWidth}{\textwidth-4.5ex}
    \def\boxheading{#1}
    \begin{lrbox}{\GrayRoundedBox}
       \begin{minipage}{\RoundedBoxWidth}%
   }{%
       \end{minipage}
    \end{lrbox}%
    \begin{tightcenter}%
    \begin{tikzpicture}%
       \node(Text)[draw=black!20,fill=white,rounded corners,%
             inner sep=2ex,text width=\RoundedBoxWidth]%
             {\usebox{\GrayRoundedBox}};
        \coordinate(x) at (current bounding box.north west);
        \node [draw=white,rectangle,inner sep=3pt,anchor=north west,fill=white] 
        at ($(x)+(6pt,.75em)$) {\boxheading};
    \end{tikzpicture}
    \end{tightcenter}\vspace{0pt}%
    \ignorespacesafterend
}

\usepackage[longend,vlined,plain,linesnumbered]{algorithm2e}

\SetKwInput{KwInput}{Input}
\SetKwInput{KwOutput}{Output}
\SetKwFor{RepTimes}{repeat}{times}{end}

\let\oldnl\nl
\newcommand{\nlnonumber}{\renewcommand{\nl}{\let\nl\oldnl}}

\newtheorem{fact}[lemma]{Fact}

\renewcommand{\Pr}{\operatorname{{\bf Pr}}}

\newcommand{\bH}{\mathbf{H}}

\newcommand{\bK}{\mathbf{K}}
\newcommand{\bL}{\mathbf{L}}

\newcommand{\bP}{\mathbf{P}}

\newcommand{\bS}{\mathbf{S}}
\newcommand{\bT}{\mathbf{T}}

\newcommand{\bu}{\boldsymbol{u}}
\newcommand{\bv}{\boldsymbol{v}}
\newcommand{\bw}{\boldsymbol{w}}

\newcommand{\balpha}{\boldsymbol{\alpha}}

\newcommand{\bpi}{\boldsymbol{\pi}}

\newcommand{\bchi}{\boldsymbol{\chi}}

\newcommand{\bPsi}{\boldsymbol{\Psi}}

\newcommand{\calA}{\mathcal{A}}
\newcommand{\calB}{\mathcal{B}}
\newcommand{\calC}{\mathcal{C}}
\newcommand{\calD}{\mathcal{D}}
\newcommand{\calE}{\mathcal{E}}
\newcommand{\calF}{\mathcal{F}}

\newcommand{\calH}{\mathcal{H}}

\newcommand{\calK}{\mathcal{K}}

\newcommand{\calP}{\mathcal{P}}
\newcommand{\calQ}{\mathcal{Q}}

\newcommand{\calbE}{\boldsymbol{\calE}}
\newcommand{\calbB}{\boldsymbol{\calB}}

\title{A characterization of one-sided error testable graph properties in bounded degeneracy graphs}

\author{Oded Lachish\thanks{Birkbeck, University of London, UK. Email: \href{mailto:o.lachish@bbk.ac.uk}{o.lachish@bbk.ac.uk}.} \and  Amit Levi\thanks{University of Haifa, Israel. Email: \href{mailto:alevi@cs.haifa.ac.il}{alevi@cs.haifa.ac.il}.} \and Ilan Newman\thanks{University of Haifa, Israel. Email: \href{mailto:ilan@cs.haifa.ac.il}{ilan@cs.haifa.ac.il}.} \and Felix Reidl\thanks{Birkbeck, University of London, UK. Email: \href{mailto:f.reidl@bbk.ac.uk}{f.reidl@bbk.ac.uk}.}}
\date{}

\begin{document}

\maketitle
\begin{abstract}
We consider graph property testing in $p$-degenerate graphs under the random neighbor oracle model (Czumaj and Sohler, FOCS 2019). In this framework, a tester explores a graph by sampling uniform neighbors of vertices, and a property is testable with one-sided error if its query complexity is independent of the graph size.  It is known that one-sided error testable properties for minor-closed families are exactly those that can be defined by forbidden subgraphs of bounded size. However, the much broader class of $p$-degenerate graphs allows for high-degree ``hubs" that can structurally hide forbidden subgraphs from local exploration. 

In this work, we provide a complete structural characterization of all properties testable with one-sided error in $p$-degenerate graphs. We show that testability is determined by how copies of the forbidden graphs can share their high-degree vertices: $\calF$-freeness is testable if and only if, for every member $H$, every set of roles of $H$ that may be played by high-degree vertices, and every pattern of dependencies between these roles, some member of $\calF$ can be assembled from pieces of copies of $H$ glued according to this pattern. Our results define the exact structural boundary for testability under these constraints, accounting for both the connectivity of individual forbidden subgraphs and the collective behavior of the properties they define.
\end{abstract}
\thispagestyle{empty}
\newpage

\setcounter{page}{1}
\section{Introduction}\label{sec:intro}
The framework of graph property testing centers on randomized algorithms that, given query access to an input graph $G$, must distinguish between graphs satisfying a property $\calP$ and those that are $\eps$-far from it. A graph is considered $\eps$-far if an $\eps$-fraction of its representation must be modified to satisfy $\calP$. The primary objective is to design algorithms whose query complexity is independent of the size of the graph. If the tester accepts any input satisfying the property with probability $1$, it is said to have \emph{one-sided} error.

The field was initiated by the seminal work of Goldreich, Goldwasser, and Ron~\cite{GGR98} in the dense graph model. In this regime, graphs are represented by adjacency matrices, and testers can query any potential edge $(u,v)$. This line of research eventually led to a profound characterization of all testable graph properties in the dense graph model~\cite{alon2006combinatorial,alon2005every,alon2008characterization}, both in general (with two-sided error) and for one-sided error testers.

Another model that has received significant attention is the bounded-degree model, initially formulated in~\cite{goldreichBndDeg1997}. Several significant results have been directed toward characterizing all testable properties within this model~\cite{FichtenbergerPS19,IKN20,AdlerKP24}. However, the question of a full characterization of all testable properties in this model remains open.

Beyond the dense and bounded-degree regimes, graph property testing has been extensively studied over the past two decades in the general graph model, and more specifically for sparse  graphs~\cite{goldreichBndDeg1997,parnas2002testing,kaufman2004tight,alonTriangle2008,newman2011every,kusumoto2014testing,babu2016every,CS19,edenCycles2024,ELRR25}. There are several distinct models for the sparse graph regime, varying with respect to the allowed query types. The general graph model enables the study of properties in graphs with a bounded \emph{average} degree. However, this strengthening comes at a price: even simple properties such as triangle-freeness~\cite{alonTriangle2008} and bipartiteness are no longer testable with a constant number of queries; rather, they require a query complexity that scales as a fractional power of the graph's size. By contrast, both of these properties are testable in the dense- and bounded-degree models.

These results indicated that the approach for this model should inherently differ from that of the dense and bounded-degree models. One direction was to restrict the family of input graphs, leading to new results that we discuss later. Czumaj and Sohler~\cite{CS19}, who provided a characterization of the testable graph properties for minor-free graphs in the \emph{random neighbor} oracle model, took a major step further. In this framework, the tester interacts with the graph by querying a vertex $v$ and receiving a neighbor $\bu$ selected uniformly at random. This demonstrated that in a model slightly weaker than the general graph model, full characterizations are possible when the input is a member of the restricted family of minor-free graphs. This is to be contrasted with the fact that  characterization of the testable graph properties (one- or two sided error) is not known, even for planar graphs.

It is important to note that Czumaj and Sohler's characterization of testable properties in this model is inherently tied to the nature of one-sided error. A one-sided error tester must accept any graph that satisfies the property $\calP$ with probability $1$. Consequently, it can reject an input graph only if it finds evidence of a violation: a subgraph  in the graph that  explicitly forbid it from satisfying $\calP$. Because the random neighbor oracle can only confirm the existence of edges, and never their absence, except for the case the vertex is isolated, it is impossible for a constant-query tester to certify that a graph is, for instance, an induced cycle or a complete bipartite graph. This leads to the fundamental observation that every one-sided testable property in this model is equivalent to a monotone property characterized by a finite family of forbidden subgraphs $\calF$ (except again, with the exception that an isolated vertex is also forbidden, in which case, the property is non-monotone). If a graph is $\eps$-far from being $\calF$-free, the tester must be able to find a copy of at least one $H\in \calF$ with high probability.

These observations gave rise to the main tool that Czumaj and Sohler used for their result: proving that $H$-freeness is testable for every graph $H$ in the family of interest. Our paper focuses on understanding what happens in much larger graph families where we know that $H$-freeness is \emph{not} testable for every graph in the family. The goal of this line of research is to discover the tools required to characterize the testable properties of natural graph families in the random neighbor oracle model. 

In this paper, we study properties of $p$-degenerate graphs for a constant $p$ (also known as $(p,1)$-admissible graphs, and bounded-arboricity graphs). A graph is $p$-degenerate if its vertices can be ordered so that every vertex has at most $p$ neighbors preceding it in the ordering. Note that while the average degree of a $p$-degenerate graph is bounded, such a graph may still contain vertices of  arbitrarily high degree. The structural definition of $p$-degenerate graphs is robust, as it admits several equivalent characterizations. In particular, it is related to \emph{bounded arboricity}, where the arboricity of the graph is $\Theta(p)$. $p$-degenerate graphs significantly generalize planar graphs and all minor-free classes,\footnote{The class of bounded-arboricity graphs contains an infinite nested collection of subfamilies---the family of $(p,r)$-admissible graphs, $r \in \mathbb{N}$---each of which contains all minor-free graphs.} providing a rich landscape for studying properties in graphs with unbounded degrees. Indeed, much of the recent study on graph property testing and local algorithms has focused on this family~\cite{ERR19, ERS20,TrianglesArboricity,EMR22,ERR22, edenCycles2024}.

\subsection{Our results}
In this work, we provide a complete structural characterization of the one-sided error testable properties for bounded-degeneracy graphs under the random neighbor oracle model. 
The logical progression of our proof proceeds as follows:
\begin{itemize}
\item \textbf{Reduction to a forbidden family:} It is standard to note that any one-sided error testable graph property is equivalent to being $\calH$-free for a finite family,  $\calH$, of constant size forbidden subgraphs. Hence, we only consider testing $\calH$-freeness, for a finite fixed set of connected\footnote{it is eeasy to see that testing for freeness of non-connected graph can be reduce to testing each component by itself.} forbidden graphs $\calH$ that does not contain an isolated vertex. 
    \item \textbf{Characterization of $H$-freeness:} We start with $\calH = \{H\}$, namely a unique forbidden graph $H$, and characterize those $H$'s that are $2$-connected and for which $H$-freeness is one-sided error testable in the random neighbor model. The characterization is essentially about the connectedness of $H$. It turns out however, that for non-$2$-connected forbidden graphs $H$, the situation can be much more complicated, depending on how the blocks of $H$ are connected. 
    \item To complete the characterization 
    on not-necessarily $2$-connected forbidden graphs, as well as the case of family of forbidden graph $\calF$ that might contain non-testable members, we need some extra machinery. This machinery is a 'characterization' of some features corresponding to the way the appearances of a forbidden graph $H$ share vertices in a  possible input graph $G$ that is far from being $H$-free. Understanding the corresponding collection of such features implies the full characterization.
\end{itemize}

\subsection{Technical Overview}\label{sec:technical_overview}
We begin with the standard observation that any tester in the random neighbor oracle model can be simplified to a canonical tester. This tester selects a set of initial nodes at random and initiates a local graph exploration (of constant depth) from each. Its decision to accept or reject is based solely on the subgraph discovered during this exploration (see Section~\ref{sec:caninical}). Crucially, a tester in this model only knows for sure that there exist edges between a queried vertex and a random vertex returned by the oracle; 
it cannot verify the non-existence of an edge in this model with one exception - the detection that a query vertex is isolated. Hence a one-sided error tester can only reject if it discovers a set of edges that form a forbidden subgraph, with the exception below. Therefore,  a property is one-sided error testable in this model only if it is $\calH$-free for some fixed size set $\calH$ of forbidden graphs. 

{\em Remark on including an isolated vertex in $\calH$: } There are non-monotone properties that are testable with one-sided error. This occurs only when we include an isolated vertex as a forbidden subgraph. The reason is
that while a tester can see only a subgraph, it can – and this is the only exception, learn that a
vertex is isolated.  Hence, simple property: ``having no  isolated vertex" is non-monotone and one-sided  
error testable. A more interesting property is being a perfect matching. It is not
downward monotone, and again testable. The tester rejects iff it sees an isolated vertex, or a vertex with two
distinct neighbors. The correctness is easy a left for the reader.

Hence, for the following, our family of forbidden subgraphs never includes an isolated vertex. Hence $1$-sided error are monotone.

We start with the relatively simple characterization of properties defined by a single forbidden subgraph $H$ that is $2$-connected. So, for a fixed graph $H$, we say that $H$ is testable if $H$-freeness is testable for every bounded degeneracy graph, and otherwise we say that $H$ is not testable.

\begin{theorem}\label{thm:thm1}  Let $H$ be a 
  $2$-connected graph. $H$-freeness is one-sided error testable,  if and only if for every independent set $S \subseteq  V(H),$ the subgraph induced by  $V(H)\setminus S$ is connected.
\end{theorem}

In order to better understand our necessary condition, it is illuminating to consider the special case where $H$ is a labeled $C_4=\{a,b,c,d\}$ with separator $S=\{a,c\}$ that is an independent set. The reason why $H$ is not testable follows from the following distribution on  ``hard-to-test" graphs. We first construct  the following  $2$-degenerate graph $G,$ on $2n + 2 \sqrt{n}$ vertices.  $A$ and $C$ are two disjoint sets of $\sqrt{n}$ vertices of high degree, called ``hubs". In addition $L_1,L_2$ are two $n$-sized disjoint sets of degree $2$ vertices.  Each vertex $v \in L_1$ is connected to a unique pair of \emph{hubs} $(i,j) \in A \times C$. Similarly, the same is done for each $u \in L_2$.  Hence, every pair $(i,j)\in A\times C$ is connected by a $2$-path through $L_1$ and another $2$-path through $L_2$ forming a copy of $C_4$. Altogether, the graph $G$ is  $2$-degenerate (with the order of $A,C$ first, followed by $L_1 \cup L_2$). Further, $G$ contains $n$ edge-disjoint copies of $C_4$ (one for each unique pair of hubs). This ensures that the graph is $\Omega(1)$-far from being $C_4$-free. 

The difficulty of finding a $C_4$ where the vertices are randomly permuted is due to the fact that in order to find a $C_4$-copy one needs to find a specific pair $(i,j) \in A \times C$, and the two corresponding $u \in L_1, ~ v \in L_2$ that are a ``matched" pair, each connected to the same $i,j$. 
Because the oracle returns a random neighbor, a tester querying a hub vertex $i$ is essentially pulling a ``random ticket" for a path to some hub $j$.
The core of the lower bound lies in the independence of $L_1$ and $L_2$ and the large degree of the vertices of $A$ and $C$. The probability of finding such a match $u,v$ is vanishingly small (with respect to $n$). Specifically, after $\ell$ queries, we show that the probability of finding a $C_4$ is only $O(\ell^2/\sqrt{n})$ - this is essentially by the birthday-paradox argument. Hence,  a lower bound of $\Omega(n^{1/4})$ is obtained for the number of queries. This argument is generalized to every $H$ for which there is a separating independent set $S$ as stated in the theorem.

To prove the sufficient condition for testability, we note that if the graph $G$ has degree bound $h$ for an arbitrary constant $h$, then $H$-freeness is testable (for any $H$)\footnote{this is quite standard - if $G$ is bounded degree, then sampling a vertex $v$, $v$ will be in a $H$-appearance w.h.p. Then, running a depth $|H|$ BFS starting from $v$ will find this $H$-appearance. }. Hence, we define for an input graph $G$, the set of heavy vertices; these are vertices of degree higher than some suitable constant $h$. We then employ a two-stage ``cleaning" procedure. Since a one-sided tester only rejects upon finding an explicit copy of $H$, we can conceptually ``remove" edges that are difficult to sample without significantly changing the graph's distance from being $H$-free (if the original graph $G$ is $\eps$-far from $H$-freeness, the resulting graph remains $\Omega(\eps)$-far).

\begin{itemize}
    \item \textbf{Degree-based cleaning:} We first remove all edges between high-degree vertices (``heavy" nodes). This edge set is has a small size due to the degeneracy property of the graph. Hence this leaves the graph far from being $H$-free.
    \item \textbf{Density cleaning:} We further refine the graph to ensure that any heavy vertex that participates in an $H$-copy actually participates in a large number of such copies. This step ensures that the random neighbor oracle has a non-negligible probability of sampling a ``useful" neighbor from a random heavy vertex.
\end{itemize}
The above conceptual process allows us to prove that if a $2$-connected $H$ has no independent sets that are separating (unlike a $C_4$), and $G$ is far from $H$-freeness, then at least one of the following is true: (a) $G$ contains many $H$-copies in which all of its vertices are non-heavy; in this case, finding an $H$-copy is easy, as in a bounded degree graph; (b) $G$ contains many $H$-copies in which there are some heavy vertices, but that are not separating --- then, again as in the bounded degree case, an $H$ copy can be found by making a suitable BFS induced only on non-heavy vertices. Such a search will find an $H$ copy, as  $H$ is not separated by the heavy vertices that it might contain. 

Next we move to $H$-freeness for non necessarily $2$-connected graphs. A first  simple observation is that if every 2-connected block of $H$ is testable for freeness then so is $H$-freeness. However, this is not a necessary condition (an example is discussed in Section~\ref{sec:further}).

An analogous fact is that if for every $H \in \calF$, $H$-freeness is testable, then so is $\calF$. Again, while this is sufficient  this is, somewhat surprisingly, not a necessary condition.  
As a motivating example (discussed further in Section~\ref{sec:family}), consider the property of forbidding  the $4$-cycle $C_4$, and the star  with $10$ leaves $\mathrm{ST}_{10}$. While $C_4$-freeness is non-testable in general $p$-degenerate graphs~\cite{edenCycles2024}, the property of  $\calH=\{C_4,\mathrm{ST}_{10}\}$-freeness  is testable. The inclusion of the star in the forbidden set ensures that any graph satisfying the property has a maximum degree of at most $9$. In this bounded-degree regime, the structural ``blind spots" that would normally hide $C_4$ are prohibited, rendering the collective property to be testable.

The example above is a rather simple example of a more general phenomenon. The non-testability of a graph $H_1\in \calF$ does not necessarily ``doom" the testability of a family $\calF$ containing it. The inclusion of additional forbidden subgraphs $\{H_2,\ldots, H_\ell\}$ can ``rescue" the property by fundamentally altering the structural regime in which the tester operates. The same is true for a non-testable block $H_1$ of $H$ where other blocks might be testable.

In both cases above the additional possible testability is implied by the way the non-testable graph $H_1$ (or a non-testable block $H_1$ of $H$), appears in an input graph $G$ that is far from $H_1$-freeness. This forces us to look deeper at the interaction between the forbidden subgraphs and the specific input graphs that attempt to hide them. This leads to our core result:

     \paragraph{Characterization of testable families of forbidden graphs.} Testing $\calF$-freeness may be possible even if $\calF$ contains a graph $H_1$ for which $H_1$-freeness is not testable. To capture this, we consider, for every $H\in\calF$, every independent set $I\subseteq V(H)$ of roles played by high-degree vertices, and every pattern $D$ of dependencies between these roles, the $(H,I,D)$-generalized cacti: the graphs that a local exploration may see in an input graph whose copies of $H$ share their high-degree vertices according to $D$. Our main result (Theorem~\ref{thm:main}) states that $\calF$-freeness is testable if and only if, for every such triple $(H,I,D)$, some member of $\calF$ is a subgraph of an $(H,I,D)$-generalized cactus. There are only finitely many triples, and the condition depends only on $\calF$. The same theorem settles the case of a single forbidden graph $H$ that is not $2$-connected.

\vspace{-0.3cm}
\subsection{Related work}
Testing properties of sparse graph has been studied extensively under different models of access. In the \emph{general graph} model~\cite{parnas2002testing,kaufman2004tight,alonTriangle2008}, the algorithm can only use queries of these three types: degree query where for an input $v$ the oracle returns $\deg(v)$, neighbor query where for an input $v$ and $i\in[n]$ it returns the $i$-th neighbor (or a special character $\bot$ if $\deg(v)<i$) and pair query $(u,v)$ on which the return is whether the edge $(u,v)$ exists. Many natural graph properties have been proven to be non-testable in this model, and some exhibit a huge gap between one-sided error and two-sided error algorithms~\cite{goldreichBndDeg1997,alonTriangle2008, kaufman2004tight, newman2011every,edenCycles2024}. 
We still seem to be far from having a characterization of the testable graph properties in general. This is also true  if we restrict the testers to be one-sided error.

A less general model (in terms of relevant input graphs) that has received significant attention  is the \emph{bounded degree} model~\cite{goldreichBndDeg1997}. In this model, we have the additional restriction that the degree of the graph is bounded by a predefined constant $d$. The oracle access in this model is as before, although it is easy to observe that allowing $O(1)$-overhead, the weakest random-neighbor query can simulate every query type as above. As for the general model, we seem to be very far from understanding which graph properties are testable. One reason is that this model contains expander graphs, which are notoriously hard to test for many properties, as the local view around a vertex  may be just a tree.  A relatively general result characterizing the monotone and hereditary one-sided error testable properties appears in \cite{IKN20}. Another general result for this model is when the input graph is restricted to be planar. \cite{CSS09} show that in this case any hereditary graph property is testable. This approach can be generalized to any class of graphs that can be partitioned into constant-size components by removing a small number of  edges of the graph. Graphs satisfying this property are called \emph{hyperfinite}, and they include all bounded-degree minor-closed graph families. A sequence of works~\cite{benjamini2008every,hassidim2009local,newman2011every} culminated in the result that all hyperfinite properties are testable. Using similar methods, ~\cite{ito2015every} shows that every property of a certain class of scale-free multigraphs is testable, and recently some more general results were obtained for simpler subclasses of planar graphs: for trees in~\cite{kusumoto2014testing} and outerplanar graphs in~\cite{babu2016every}. Some other results for specific properties can be found in~\cite{goldreichBndDeg1997,ADPR03,nguyen2008constant,alonTriangle2008}, and others. 

Other areas of sparse graph properties that were studied extensively, as for property testing per-se, and in the general local-algorithms and approximate counting is when the input graphs are restricted to be locally sparse.  As explained above, one such very general family of graphs is the $O(1)$-degenerate graphs, or equivalently, bounded arboricity graphs. However, as noted above, characterization of the testable graph properties is not known for these families of graphs (under the general query-model), in any of the different models, and even for one-sided error testing.


  The other relevant results, in view of the non-testability of $C_4$-freeness for  general $O(1)$-degenerate graphs~\cite{edenCycles2024}, but contrasted with the fact that any $H$-freeness is testable for  minor free graphs~\cite{CS19}, is that of~\cite{H-freeness,AwofesoGLLR25,awofeso2025sufficient,KMPV26}. They consider a hierarchy of sparse graph families by using the measure of $r$-admissibility~\cite{NdM12}. The family of $O(1)$-degenerate graphs is exactly that of $1$-admissible graphs. The parameter of $r$-admissibility partitions the set of $O(1)$-degenerate graphs into an infinite nested families of graphs, each of which contains minor-free graphs. As the family becomes smaller (being $r$-admissible with larger $r$) $H$-freeness can be tested for larger collection of forbidden subgraphs $H$. As stated before, one of the main results in this work is the characterization of these forbidden graphs $H$ for which $H$-freeness is testable for the $1$-admissible graphs.

\section{Preliminaries and Notations}\label{sec:prelim}
We denote $[k]=\{1,\ldots,k\}$. We use boldface letters such as $\bv$ to denote random variables. All graphs in this work are simple, undirected.
For a graph $G=(V,E)$ and $v\in V$, we use $\deg(v)$ to denote the degree of $v$, and $N(v)$ as the set of neighbors of $v$. For $V'\subseteq V$ we use $G[V']$ to denote the subgraph of $G$ induced by $V'$ and $N(V') = \{u \in  N(v), ~ v \in V'\}$.

\begin{definition}[\emph{2-block}] Let $H=(V(H),E(H))$ be a connected graph. A subgraph $H'$ of $H$ is called a \emph{2-block} of $H$ if it is a maximal $2$-connected subgraph of $H$\footnote{Note that a single edge is considered $2$-connected block. }.
\end{definition}
The following is well known, 
\begin{fact}[block decomposition] \label{lem:2-block-decomp}For every connected graph $H$, there exists a \emph{2-block decomposition} $\calB=(H_1,\ldots,H_\ell)$ such that $V(H)=\bigcup_{j\in [\ell]}V(H_j)$, every $H_j$ is a 2-block and for every $i\neq j\in[\ell]$, $|V(H_i)\cap V(H_j)|\leq 1$.   $v \in V(H_i)\cap V(H_j)$ is a separation point, also called \emph{articulation point}.
\end{fact}

%
\begin{definition}[components]
    Let $G=(V,E)$ be a graph,  $S \subseteq V$ a separating set, and $C=(V',E')$ a component of $G \setminus S$.  We denote by $C(S)$ the connected induced subgraph of $G[V' \cup S]$ that contains $C$. Note that $C(S)$ may not contain all $S$-vertices. We call $C(S)$ an $S$-component.
\end{definition}





\begin{definition}[$p$-degenerate] For $p\in \mathbb{N}$, a graph $G=(V,E)$ is \emph{$p$-degenerate} if every non-empty subgraph of $G$, contains a vertex of degree at most $p$. In particular, any subgraph of size $k$ of $G$ can have at most $kp$ edges. We use $\calA_p$ to denote the family of $p$-degenerate graphs.\footnote{The family $\calA_p$ is also known as $(p,1)$-admissible graphs.}
\end{definition}

\begin{definition}[$H$-appearance]
A subgraph $G'$ of  $G$ is an $H$\emph{-appearance} if $G'$  is isomorphic to $H$.    
\end{definition}

There are several equivalent and related notions to that of being $p$-degenerate. Most importantly, a graph is $p$-degenerate if and only if its arboricity is bounded. Specifically, the arboricity $\alpha(G)$ of a $p$-degenerate graph satisfies $\lceil(p+1)/2\rceil\le\alpha(G)\le p$. Furthermore, $p$-degeneracy is a specific case of the admissibility hierarchy~\cite{NdM12}. In this context, $p$-degenerate graphs are exactly those that are $p$-bounded $1$-admissible (usually denoted by $(p,1)$-admissible).

\noindent
\vspace{0.2cm}
{\bf Property testing and forbidden subgraphs freeness.}\\
A \emph{graph property} $\calP$ is a family of graphs closed under isomorphism. 
For $\eps\in[0,1]$ we say that a graph $G=(V,E)$ is \emph{$\eps$-far from} $\calP$ if one has to modify more than $\eps|V|$ edges from $G$ to obtain a graph satisfying $\calP$, and otherwise we say that it is \emph{$\eps$-close}\footnote{While the distance is typically defined by the fraction of edge modifications required to satisfy property $\calP$, the linear edge bound of $p$ degenerate graphs ($|E|\le p|V|$) ensures that definitions based on the number of vertices versus the number of edges are equivalent up to a factor of $p$.} to $\calP$.

\begin{definition} A $q$-query \emph{property tester} for a graph property $\calP$, and graphs of $n$ vertices, is a randomized algorithm that receives as input parameters $\eps>0$ and random neighbor oracle access to a graph $G$ with $n$ vertices. The algorithm makes at most $q$ random neighbor queries to the input and satisfies the following. If the graph is $\eps$-far from $\calP$, then the algorithm rejects with probability at least $2/3$. If $G\in\calP$, then the algorithm accepts with probability at least $2/3$. The tester has {\em one-sided error} if it accepts every $G \in \calP$ with probability $1$, and otherwise it has a \emph{two-sided error}.
\end{definition}

{\bf comment: } The framework of property testing is non-uniform in nature. The tester is designed for a given $n$, and hence a tester is formally an infinite collection of testers, one for each $n$. Like-wise, while in the definition above, the tester gets $\epsilon$ as an input, we may assume that each individual tester is designed for a given $\epsilon$.  We did define $\epsilon$ as an input parameter as in most testers, the distance parameter $\epsilon$ can indeed treated as an input (namely, the same tester would work, adapting to the actual value of $\epsilon$).

Our focus in this work is on {one-sided error} testers in the {random-neighbor} model. In this model, the oracle access to the graph is by the {random neighbor oracle}, where an algorithm may query any vertex $v\in V(G)$ and the oracle returns a vertex $\bu$ chosen uniformly at random from the set of neighbors of $v$ in $G$. Thus, for non-bounded degree graphs, a degree of a vertex, as well as the existence or non-existence of an edge between a pair of vertices cannot be determined in constant time.  An $O(1)$-query tester can only discover a constant size subgraph of the input graphs and form its decision only on this subgraph. A one-sided error tester in the model can thus reject only on discovering a forbidden subgraph, and hence can accept only downwards monotone properties (recall, this is under the assuption that an isolated vertex is not forbidden).  

We say that a property $\calP$ is \emph{testable} if it has a tester with query complexity which is independent on the graph size $n$, but may depend  on $\eps$ and possibly some structural parameters of the input graphs (such as the degeneracy $p$). We sometimes use the notion of an $\eps$-test, referring to a (non-uniform) property tester in which the proximity parameter $\eps$ is fixed in advance, as opposed to the standard uniform setting where $\eps$ is given as input to the tester.

{\bf Canonical Testers} \label{sec:caninical}
The following result from~\cite{CS19} provides a canonical way of describing any tester in the random neighbor oracle model.

To analyze the local structure around a vertex $v$, we use a specialized \textsf{Bounded-BFS} subroutine. The procedure explores the graph up to a fixed depth $t$ using a random-neighbor oracle. To remain query-efficient, the algorithm uses a fixed sampling parameter $s$: at each vertex encountered, the search probes exactly $s$ neighbors chosen uniformly at random.
The value of $s$ is chosen specifically to ensure that if a vertex is ``light" (its degree is below a certain threshold $h$), the algorithm will have sampled all of its neighbors with high probability. This allows the search to behave like a standard BFS for low-degree regions while providing a representative sparse sample of higher-degree neighborhoods.
A formal description of the \textsf{Bounded-BFS} procedure and the derivation of the sampling parameter $s$ are provided in Appendix~\ref{sec:BFS}.

\begin{theorem} Let $\calP=(\calP_n)_{n\in \mathbb{N}}$ be a graph property that can be tested in the random neighbor oracle model with $q$ queries and error probability at most $1/3$. Then, for every $\eps>0$, there exists $q'=\Theta(q)$ and a sequence $\calQ=(\calQ_n)_{n\in \mathbb{N}}$ such that for every $n\in \mathbb{N}$, $\calQ_n$ is a set of bounded $q$-size graphs.
The property $\calP_n$ (of $n$ vertex graphs) can be tested with error probability at most $1/3$ by the following \emph{canonical tester} that uses $q^{O(q)}$ queries: 

$(1)$ Sample a multi-set $\bS$ of $q'$ vertices uniformly at random.
    $(2)$ For each sampled vertex $\bv\in \bS$ run \textsf{Bounded-BFS}$(G,v,q',q')$ and obtain the explored subgraph $G'_{\bv}$.
    $(3)$ If the union of explored $\{G'_{\bv}\}_{\bv\in \bS}$ contains an element $Q\in \calQ_n$, then the tester rejects and otherwise it accepts. 
    
    Additionally, if $\calP$ can be tested with one sided error in the random neighbor oracle model, then the canonical tester for $\calP$ also has one-sided error. We refer to such a tester as a \emph{$q'$-canonical tester}.
\end{theorem}

The above shows that, without loss of generality, one can assume that any testable property can be tested by a canonical tester with constant (independent of the input size) query complexity. In particular, this implies that for one-sided error testers, the test can reject only upon discovering a forbidden subgraph. This implies the following definitions. 

We say that a graph $G$ is \emph{$H$-free} if $G$ does not contain $H$ as a subgraph and denote by $\calP_H$ the family of all such graphs. Correspondingly, $G$ is said to be $\eps$-far from being $H$-free (alternatively, $H$-freeness) if more than $\eps|V(G)|$ edges must be deleted from $G$ in order to make it $H$-free.

These definitions extend naturally to families of forbidden graphs. If $\calH$ is a {\em finite} family of finite graphs, $G$ is $\calH$-free if it is $H$-free for every $H\in \calH$. We denote by $\calP_\calH$ the property of all graphs that are $\calH$-free, and define the distance analogously. We note that $\calP_\calH$ is a downwards monotone graph property for any collection $\calH$. Moreover, every downwards monotone property can be described by $\calP_\calH$ for some (possibly infinite) collection of forbidden graphs $\calH$.

It is standard (see also \cite{CS19,czumaj2020testable}) that only properties of the form 
$\mathcal{P}_{\mathcal{H}}$ for
a  finite collection of bounded size forbidden graphs $\mathcal{H}$ are  
 one-sided error testable 
 in the random-neighbor model (again with the exception that the property allows isolated vertices).  
 
 Our goal is to characterize these  families $\calP_{\calH}$ for which $\calP_{\calH}$-freeness is one-sided
error testable. As a one-sided error test rejects only if it finds an $H$-appearance in the input graph $G$, for $H \in \calP_{\calH}$, we may assume in what follows that the input graph $G$ is always $\epsilon$-far from being $\calP_{\calH}$-free, and will characterize these family for which such a graph $H \in \calP_{\calH}$ can be found with $\Omega(1)$ probability.

The following standard fact relates $G$ being $\epsilon$-far from being $\mathcal{P}_{\mathcal{H}}$-free in the aforementioned Hamming metric to the presence of a linear number of edge-disjoint forbidden subgraphs in $G$.

\begin{lemma}
\label{lemma:edge-disjoint-app}
 Let $\mathcal{H}$ be a $r$-collection of $k$-size graphs. If graph $G$ is $\epsilon$-far from $\mathcal{H}$-free 
 then it has an edge-disjoint collection $\mathcal{H} (G)$, 
 of size at least $\eps n/krp$ of $H$-appearances, for some 
 $H \in {\calH}$ 
 (note that if $k\ge p$, then we have at least $\eps n/rk^2$ such appearances).
\end{lemma}
\begin{proof}
 Let $\mathcal{C}$ be any maximal collection of edge-disjoint appearances of $\calH$-members in $G$. Removing all edges in every such appearance results in a $H$-free graph, and the lemma follows.
\end{proof}

{\bf An Important Note:} In our characterization we think of the family of input graphs as $p$-degenerate for some $p=O(1)$. Namely, if we say that $\calH$-freeness is testable for $p$-degenerate graphs, we mean that it so for any $p=O(1)$, with the test knowing $p$ (and in particular, its query complexity may depend on $p)$. On the other hand, if $\calH$-freeness is not testable, that means that for every $p=O(1),~ p \geq p^*$, for some constant $p^*$ that may depends on $\calH$, no algorithm of constant query complexity can test $\calH$-freeness against every $p$-degenerate graph.

\section{Warm-up: a single $2$-connected forbidden graph}
\subsection{Heavy and light vertices}\label{sec:H-free}

Let $\calF$ be a fixed set of forbidden graphs, each of size at most $k$. We present here some crucial definitions and observations regarding testing $\calF$-freeness. 

In what follows, the input graphs are 
 always  $n$-vertex $p$-degenerate graphs that have a linear size collection $\calF(G)$ of $H$-appearances, for some $H \in \calF$, as stated in Lemma~\ref{lemma:edge-disjoint-app}.

\newcommand{\Heavy}{\textrm{Heavy}}
\newcommand{\Light}{\textrm{Light}}
\begin{definition}
  [semi-bipartite structure] \label{def:semi-bipartite}
Let $G=(V,E)$.   For a natural number $h\in\mathbb{N}$ (usually a fixed integer
independent of $n$ but may be a function of $\eps$ and the graph $H$) let $\Heavy_h= \{v \in V(G)\mid
~ \deg(v) \geq h\}$ and $\Light_h= V(G)\setminus \Heavy_h$.
We say that $G$ is semi-bipartite with respect to $\Heavy_h$ if $\Heavy_h$ is an
independent set in $G$.
\end{definition}

\begin{lemma} \label{cl:bipar}
 Let $p\in \Naturals$, $\epsilon > 0$ and  $h\geq 4p^2/\eps$. 
  If $G$ is $p$-degenerate and $\eps$-far from
  a monotone graph property $\mathcal{P}$, then there exists a spanning subgraph 
$G'$ of $G$, obtained by deleting at most $\eps n/2$ edges, which is semi-bipartite with respect to $\Heavy_h$, and is $\eps/2$-far from $\mathcal{P}$. In particular $G'$ contains a collection of edge-disjoint $H$-appearances as in Lemma~\ref{lemma:edge-disjoint-app}.
  \end{lemma}
  \begin{proof}
    By the assumption that $G$ is $p$-degenerate, 
    any size $t$ subgraph of $G$ has at most $tp$ edges. 
    By averaging (using that $|E(G)| \leq pn$) we conclude that $|\Heavy_h| \leq
    2pn/h$. Hence, 
    $G[\Heavy_h]$ has at most $2p^2 n/h$ edges. Thus, since $h\geq 4p^2/\eps$, deleting these edges results in a graph $G'$ that is semi-bipartite with respect to $\Heavy_h$ and is $\eps/2$-far from $\mathcal{P}$. 
  \end{proof}

We use $G'$ only as a conceptual object in the analysis. The tester is one-sided and rejects only upon discovering a constant-size forbidden subgraph. Since $G' \subseteq G$, any forbidden subgraph found in $G'$ is also present in $G$, and therefore constitutes a valid witness for rejecting $G$.
On no-instances, although deleting edges may remove some forbidden subgraphs, our analysis shows that if $G$ is $\epsilon$-far from the property, then the resulting graph $G'$ still contains many forbidden subgraphs of the relevant type.

In view of Lemma~\ref{cl:bipar} and the discussion above, we may henceforth assume that the input graph $G$ is semi-bipartite with respect to the appropriate parameter $h$. More specifically, we state the following:

\begin{remark}\label{rem:props}
    As we consider only one-sided error tests of downward monotone properties, we may assume in what follows that all our input graphs are $\eps$-far from being $H$-free for some forbidden $k$-size graph $H$, and have the semi-bipartite property with respect to $\Heavy_h$.  In particular, the input graph $G$ has an edge-disjoint collection $\calH(G)$, of $H$-appearances of size at least $\eps n/kp$ (note that if $k\ge p$, then we have at least $\eps n/k^2$ such appearances).
\end{remark}

Our aim is to characterize the families $\calF$  for which there is a  test, under the above assumptions, that can find an $H$-appearance for $H \in \calF$,  with high probability. 
We first consider the easy case\footnote{This case is also treated in all previous studies of the bounded degree models.} in which $\calH(G)$ contains at least $\Omega(\eps n/kp)$ $H$-appearances $H'$ for which $\Heavy_h\cap V(H')=\emptyset$.

\begin{lemma}\label{lem:all-light-test} Fix $h\in\mathbb{N}$ and suppose $G$ is semi-bipartite with respect to $\Heavy_h$ and  $\eps$-far from $\calP_H$. If all the $H$-appearances in $\calH(G)$ contain only vertices in $\Light_h$, then there exists a one-sided error $q_H$-canonical $\eps$-tester for $\calP_H$, where $q_H=O_{k,p,\eps}(1)$.
\end{lemma}
\begin{proof} Since $\calH(G)$ contains $\Omega(\eps n/kp)$ edge-disjoint appearances $H'$ containing only vertices in $\Light_h$, a $q_H$-canonical tester finds a vertex $v$ in such appearance $H'$ with high probability. Conditioned on this event, and the fact that all vertices in such appearance are light, the explored subgraph centered around $v$ contains $H'$ as a subgraph, causing the tester to reject.
\end{proof}

In view of Lemma~\ref{lem:all-light-test}, we assume in what follows that our input graphs $G$ have the property that every $H$-appearance in $\calH(G)$ contains a $\Heavy_h$ vertex.

\subsection{The 2-connected case}\label{sec:2-conn-main}

We start with the following simpler property of being $H$-free for a fixed $2$-connected graph $H$ of (constant) size $k$. We give an explicit characterization of when for such $H$, the property $H$-freeness is testable. This characterization will follow from the much more general characterization, but less explicit,  in Theorem~\ref{thm:main}, but it will be instructive to see it first.

 The following definition is crucial.
\begin{definition}[obstacle]\label{def:obstacle}
    Let $H$ be a $2$-connected graph, $S\subseteq V(H)$ an independent set in $H$.  If $S$ is a minimal-separating set, we call $S$ an {\em obstacle} for $H$.
\end{definition} 
The following is a restatement of Theorem~\ref{thm:thm1}.
\begin{theorem}
  \label{thm:2}
  $\mathcal{P}_H$ is one-sided error testable for a
  $2$-connected graph $H$ if and only if $H$ has no obstacles.
\end{theorem}

To establish necessary condition of Theorem~\ref{thm:2}, we define a distribution of $p$-degenerate graphs that are far from $\calP_{H}$ but difficult to distinguish from $H$-free graphs. 

\begin{lemma}
  \label{lem:lb}
  Let $H$ be a $2$-connected  graph, $S\subseteq V(H)$ an obstacle, $|S|=r$. Then there is a distribution
  $\calD$ on $n$-vertex graphs, each being $p$-degenerate and
  $\eps$-far from $\mathcal{P}_H$, such that any one-sided error test
  making $q=o(n^{1/2r})$ queries  can find a copy of $H$ with probability at most $o(1)$.
\end{lemma}

\begin{proof} To better understand our lower bound construction, we start with the extremely simple case where $H$ is the {\em labeled}  $C_4 = (a,b,c,d)$, with obstacle $S=\{a,c\}$. 
We construct the following graph $F$ on $2m + 2m^2$ vertices as follows: Let  $A, C$ be two disjoint sets of size $m$, and two disjoint sets $L_1, L_2$ of size $m^2$. For each $(i,j) \in A \times C$ there is a unique $v=v_{i,j} \in L_1$ that is connected to both $i$ and $j$, and similarly, a unique $v'_{i,j} \in L_2$ that is connected to $i, j$. Note that for every pair $(i,j) \in A \times C$, $\{i, j, v_{i,j}, v'_{i,j}\}$ form a $C_4$ in $F$. Thus, there are $m^2$ such $C_4$-copies that are edge disjoint. 
This construction implies that $F$ is $2$-degenerate and $1/4$-far from $\mathcal{P}_H$. Now the distribution $\mathcal{D}$ is generated by permuting the vertices in $L_1$ and $L_2$ via two independent random permutations $\bpi_1, \bpi_2$ of $L_1$ and $L_2$.

To prove the lower bound, we note that each query targets either a vertex $v \in L_1 \cup L_2$ or a vertex $i \in A \cup C$. In the first case, assuming without the loss of generality that $v \in L_1$, the oracle reveals the unique neighbors $i \in A$ and $j \in C$ such that $\bpi_1(i,j) = v$. In the second case, a random neighbor query to $i \in A \cup C$ returns a vertex $v \in L_1 \cup L_2$ and its other neighbor $j$, which similarly identifies a triplet $(i, v, j)$. Since $\bpi_1$ and $\bpi_2$ are uniform and independent, the specific labels of vertices in $L_1 \cup L_2$ provide no information beyond identifying the $2$-path between $i$ and $j$ through either $L_1$ or $L_2$. Thus, we may assume each query simply discovers a pair $(i, j) \in A \times C$ connected via a vertex in one of the two sets.

Let $Q_\ell$ be the set of pairs discovered after $\ell$ queries, partitioned into $P_1$, the pairs connected via $L_1$,  and $P_2$  - the pairs connected via $L_2$. A $C_4$ is detected only if there exists some $(i, j)$ such that $v_{i, j} \in P_1$ and $v_{i, j} \in P_2$, both that are connected to the same two vertices  $i,j$, are discovered. Consider the $(\ell+1)$-th query to a neighbor of $i \in A$, and note that there are $2m$ neighbors incident to $i$: $m$ in $L_1$ and $m$ in $L_2$.
If the oracle returns a vertex $v=v_{i,j} \in L_1$, a $C_4$ is discovered only if the vertex  $v'_{i, j}$ was already discovered in $P_2$. Because $\bpi_1$ and $\bpi_2$ are independent, the mapping of neighbors in $L_1$ is unknown regardless of the vertices found in  $P_2$. At step $\ell+1$, there are at most $|P_2| \leq \ell$ such ``matching" vertices in $L_1$. Since the oracle selects from $2m$ total neighbors for vertex $i$, the probability that the query returns a neighbor completing a $C_4$ is at most $\frac{\ell}{2m}$.

The total success probability after $\ell$ queries is bounded by:
$\sum_{j=1}^{\ell} \frac{j}{2m} = \frac{\ell(\ell+1)}{4m}$.
For this probability to be $\Omega(1)$, we require $\ell = \Omega(\sqrt{m})$. Given $n \approx 2m^2$, we have $m \approx \sqrt{n/2}$, yielding a lower bound of $\ell = \Omega(n^{1/4})$. Thus, any tester making $o(n^{1/4})$ queries fails to find a $C_4$ with high probability.

Consider now the case where the $2$-size obstacle $S=\{a,b\}$, and $H \setminus S$ contains $t$ components $C_1, \ldots , C_t$ of sizes $c_r$ for $r\in[t]$. 
We will define $\mathcal{D}$ in
a similar conceptual way. We start with a fixed graph where  $A,C$ will contain $m = \Theta(\sqrt{n})$ 
nodes as before. We further will have $t$ sets $L_1, \ldots ,L_t$ each
of size $c_r\cdot m^2$ for $r\in[t]$, with the intention, as in the
very simple case above, that for each $(i,j) \in A \times C$, taking
the role of $a,b$ respectively, we will form a unique $H$-appearance
by forming $t$ vertex disjoint components, $C_r(i,j)$ isomorphic to $C_r$, $~r\in[t]$, where for $C_r$ we use  vertices from $L_r$. Thus
for each $(i,j)$ we will get a disjoint copy of $H$ except for the
vertices $i,j$. 

Again for $m = O(\sqrt{n})$ we will have $m^2$ such edge disjoint appearances on a graph of size $2m+ k m^2$. Hence the graph constructed is $\Omega(1/k)$-far from being $H$-free, in addition to being $2$-degenerate. Now, the distribution is exactly as before by permuting independently the vertices in  each $L_r$ for $r\in[t]$. In order to reduce to the previous simple case, one can consider a augmented oracle which returns the \emph{whole} component $C_\ell$ for any query made to a node in the component. By the same reasoning as above,  making $o(n^{1/4})$ queries will find a matching pair of components $C_1(i,j),C_2(i,j)$ with $o(1)$ probability.

Finally for the general case, let $S = \{v_1, \ldots v_r\}$ be a minimum size obstacle in the labeled $H$. Let $C_1, \ldots ,C_t$ be the $S$-components. We take now $r$ disjoints sets $A_1, \ldots ,A_r$ of size $m$ and $t$ disjoint sets $L_i, ~ i \in [t],$ as above, each of size $|V(C_i)| m^r$ respectively.  Then for each tuple $\bar{u}=(u_1, \ldots ,u_r) \in \prod_{i=1}^r A_i$ we form an $H$-copy by attaching the corresponding $V(C_i)$ vertices from each $L_i$ to $\bar{u}$, where these $H$ copies are vertex disjoint except for $\bigcup_{i=1}^r A_r$. See Figure~\ref{fig:lb1} for an example of an obstacle of size $3$. 

Altogether, this defines a graph $F$ on $n = \Theta (rm + km^r)$ vertices. Further, $F$ is $\Omega(1/k)$-far from $H$-free, as  any tuple $\bar{u}$ induces an $H$ copy and these copies are edge disjoint. 

Finally, we note that $F$ contains no other $H$-copies except the defined above, due to the assumption that $S$ is a minimum size obstacle. 

Now the lower bound follows, as to find an $H$ copy, one would need to find a 'collision' as described for the simpler cases above, between at least two components of $H$ sharing the same $\bar{u}$.
\end{proof}
\begin{figure}[ht!]
\centering
\begin{tikzpicture}[
  dot/.style={circle,fill=black,inner sep=0pt,minimum size=5pt},
  blob/.style={draw=blue!50!black,fill=blue!10,thick}]
\begin{scope}
  \draw[blob] plot[smooth cycle,tension=0.7] coordinates
    {(0,1) (-0.9,1.4) (-1.9,1.1) (-2.1,0) (-1.9,-1.1) (-0.9,-1.4) (0,-1) (-0.8,-0.5) (0,0) (-0.8,0.5)};
  \draw[blob] plot[smooth cycle,tension=0.7] coordinates
    {(0,1) (0.9,1.4) (1.9,1.1) (2.1,0) (1.9,-1.1) (0.9,-1.4) (0,-1) (0.8,-0.5) (0,0) (0.8,0.5)};
  \node[dot,label=above:$a$] at (0,1) {};
  \node[dot,label={[label distance=1pt]above:$b$}] at (0,0) {};
  \node[dot,label=below:$c$] at (0,-1) {};
  \node at (-1.3,0) {$C_1$};
  \node at (1.3,0) {$C_2$};
  \node[font=\Large] at (0,-2.4) {$H$};
\end{scope}
\begin{scope}[xshift=8cm,scale=1.25]
  \draw[blob] plot[smooth cycle,tension=0.7] coordinates
    {(-0.2,1) (-1.0,1.5) (-1.8,1.3) (-1.9,0) (-1.7,-1.2) (-0.9,-1.5) (0,-1) (-0.6,-0.55) (0.4,-0.1) (-0.7,0.4)};
  \draw[blob] plot[smooth cycle,tension=0.7] coordinates
    {(-0.2,1) (0.7,1.5) (1.7,1.3) (2.0,0) (1.8,-1.2) (0.9,-1.5) (0,-1) (0.9,-0.55) (0.4,-0.1) (0.5,0.45)};
  \node[dot,label={[label distance=1pt]above:$a_i$}] at (-0.2,1) {};
  \node[dot] at (0.4,-0.1) {};
  \node[font=\small] at (0.08,0.42) {$b_j$};
  \node[dot,label=below:$c_k$] at (0,-1) {};
  \node[font=\small] at (-1.15,-0.55) {$C_1(i,j,k)$};
  \node[font=\small] at (1.25,0.55) {$C_2(i,j,k)$};
  \node[dot,label=left:$a_1$] at (-2.5,1.0) {};
  \node[dot,label=left:$b_1$] at (-2.5,0.2) {};
  \node[dot,label=left:$c_1$] at (-2.5,-0.6) {};
  \node[dot,label=right:$a_m$] at (2.5,1.0) {};
  \node[dot,label=right:$b_m$] at (2.5,0.2) {};
  \node[dot,label=right:$c_m$] at (2.5,-0.6) {};
  \node[font=\Large] at (0,-1.92) {$F$};
\end{scope}
\end{tikzpicture}
\caption{The graph $H$ has a separating set $S=\{a,b,c\}$. The graph $F$ has $m^3$ copies of $H$, one for every
$(i,j,k)\in[m]^3$. In the copy $(i,j,k)$, the vertices $a,b,c$ are played by $a_i,b_j,c_k$, and its components
are $C_1(i,j,k)$ and $C_2(i,j,k)$.}\label{fig:lb1}
\end{figure}

\begin{lemma}
  \label{lem:ub}
   Assume that a $2$-connected graph $H$ has no obstacle. Then $\mathcal{P}_H$ is one-sided error testable for the family of $p$-degenerate graphs.
\end{lemma}

\begin{proof}
Let $H$ be a $k$ vertex graph as in the lemma.   Let $G=(V,E)$ be a $p$-degenerate graph that is
  $\eps$-far from $\mathcal{P}_H$. By setting $h\ge4p^2/\eps$ and using Lemma~\ref{cl:bipar}, there is a subgraph $G'$ of $G$ which is semi-bipartite with respect to $\Heavy_h$ and is $\eps/2$-far from $\calP_H$. Since $G'$ is semi-bipartite, we can assume that any $H$ appearance in $G'$ has at least one vertex of degree at most $h$. Let $U$ be the union of all such vertices. By definition, if we remove from $G'$ every edge incident on a vertex from $U$, then the total number of edges removed from $G'$ is at most $h\cdot|U|$, and the resulting graph is $H$-free. Since $G'$ is $\eps/2$-far from $\calP_H$, we have $h\cdot |U|\ge \eps n /2$, and hence $|U|\ge \eps n/2h$.

The discussion above implies the following basic test. Choose
$\bv \in V(G)$ uniformly at random and run \textsf{Bounded-BFS}$(G,\bv,k,h\log(10\cdot\frac{h^{k+2}}{\eps} )$. Reject if and only if an $H$ appearance is discovered.  Then, with probability at least $\frac{\eps }{2h}$
this vertex $\bv\in U$ (and in particular participates in an $H$ appearance). 
Since $H\setminus S$ is connected, the BFS can traverse the entire ``light" component of $H \setminus S$,  by following only edges between low-degree vertices. Because these vertices have degree at most $h$, a \textsf{Bounded-BFS} can exhaustively explore their local neighborhoods. The heavy vertices $S$ are then discovered naturally as neighbors of this light component, without the BFS ever needing to explore the neighborhood of heavy vertices\footnote{except only to conclude that a vertex is  heavy.}.

By assumption, $H \setminus S$ is
connected for every independent set $S\subseteq V(H)$, running \textsf{Bounded-BFS}$(G,\bv,k,h\log(10\cdot\frac{h^{k+2}}{\eps} ))$ finds a copy of $H$ with probability at least $9/10$. This basic test has query complexity $\tilde O(h^{k+1})$, and will find an $H$-appearance with probability $\Omega(\frac{\eps}{h})$. Repeating the 
test $O(\frac{h}{\eps})$ times will find an $H$-appearance with 
probability at least $2/3$ and 
$\tilde O(\frac{h^{k+2}}{\eps})$ queries.\end{proof} 

  
\section{Testing $\mathcal{H}$-freeness for a family $\mathcal{H}$ of forbidden subgraphs} \label{sec:family}
As discussed in Section~\ref{sec:intro}, the only one-sided error testable properties in the random neighbor model are those characterized by $\mathcal{H}$-freeness for a family of constant-sized forbidden subgraphs $\mathcal{H}$. 
While we characterized testability for a single 2-connected forbidden graph in Section~\ref{sec:2-conn-main}, this leaves two additional tasks for a complete characterization of the families $\calF$ for which $\calF$-freeness is testable. The first is to generalize the results for one member families $\calF = \{H\}$, where $H$ is not necessarily   2-connected. The other, is to generalize the results for families of larger size.

It is important to note that the testability of each individual $H\in\calF$ is sufficient for the testability of $\calF$-freeness, but it is not necessary.

Consider, e.g., the family $\mathcal{H} = \{C_4, \textrm{ST}_{10}\}$, where $ \textrm{ST}_{10}$ denotes a star with $10$ leaves. A graph is $\mathcal{H}$-free if it does not contain a $C_4$ and has a maximum degree of at most $9$. While $C_4$-freeness is not testable on its own (as evidenced by the $2$-degenerate graph construction in Section~\ref{sec:2-conn-main}), the combined property $\mathcal{H}$-freeness is testable. 
If $G$ contains $\Omega(\epsilon n)$ vertices of degree at least $10$, a violation of $\textrm{ST}_{10}$-freeness is discovered with high probability via uniform vertex sampling. If $G$ does not contain many high-degree vertices, the graph is effectively $9$-bounded degree. In such case,  testing  $C_4$-freeness is easy since $G$ contains $\Omega(\epsilon n)$ vertex-disjoint copies of $C_4$. 

The simple example above does not fully capture the complexity of determining whether $\mathcal{H}$-freeness is testable when $\mathcal{H}$ contains a graph $H_1$ that is individually non-testable. A clear necessary condition for $\mathcal{H}$-freeness to be testable is that any graph $G$ that is $\epsilon$-far from $H_1$-freeness, and specifically chosen to be ``hard'' to test for $H_1$, must instead contain many copies of some other member $H \in \mathcal{H}$. 

This requirement forces us to characterize the family of input graphs that are hard to test with respect to a given $H_1$, but contain many copies of $H$. 

To do the above, we study the possible hard-to-test graphs with respect to a given fixed graph $H$. For this we develop the notion of dependency graph, and generalized cacti in Section~\ref{sec:cactus}. It turns out that this notion is enough to complete the characterization of both $H$-freeness where $H$ is not necessarily $2$-connected, and also $\calF$-freeness for families $\calF$ that contain more than one forbidden graph.

The remainder of this section is organized as follows. We first develop the mechanism of generalized cacti in Section~\ref{sec:cactus}. Then, in Section~\ref{sec:lb-general} we show a necessary  condition for a family $\calF$ to be testable, and  when $H$-freeness is testable for non-2-connected graph $H$. This condition is based on the collection of the relevant generalized cacti. In Section~\ref{sec:family-ub}  we complement the above and show that the relevant collection of generalized cacti determines a sufficient condition for testability of the two cases above, completing the characterization for all cases.

We note that Theorem~\ref{thm:main} is much less explicit than Theorem~\ref{thm:2}. Hence, in Section~\ref{sec:further} we discuss some instances of it where we can make a more explicit statements.
\subsection{Generalized cacti}\label{sec:cactus}
Let $H$ be a fixed connected graph.
In what follows  we characterize the fixed size subgraphs of any input graphs that a  random BFS on $G$ may see, assuming that $G$ is far from begin $H$-free.

Let $H$ be a connected graph and let $I=\{v_1,\ldots,v_r\}\subseteq V(H)$ be a \emph{labeled independent} set of $H$.
Let $C_1,\ldots,C_t$ be the components of $H-I$, and let $C_i(I)$ be the \emph{$I$-component} of $C_i$, that is, the
subgraph of $H$ induced by $V(C_i)$ and the vertices of $I$ that have a neighbor in $C_i$ (see
Section~\ref{sec:prelim}). Let $D=(I,E_D)$ be a directed graph on $I$, possibly with anti-parallel edges, that is
transitively closed. We read $v_i\to v_j$ (alternatively $(v_i,v_j)\in E_D$) as ``the vertex playing $v_i$ determines the vertex playing $v_j$''.
More precisely, consider an input graph $G$ such that
\begin{enumerate}
\item[(1)] in every $H$-appearance in $G$, the vertices playing the roles in $I$ are heavy,
\item[(2)] every heavy vertex of $G$ plays the same role of $I$ in all $H$-appearances containing it, and
\item[(3)] for every edge $(v_i,v_j)\in E_D$ and every vertex $x$ of $G$, all $H$-appearances in which $x$ plays $v_i$ use the same vertex $y$ in the role $v_j$.
\end{enumerate}

In such a graph, a random bounded-depth BFS that starts at a light vertex finds a subgraph $P$ that is an appearance
of some $C_i(I)$. If, moreover, all heavy vertices of $G$ have roles in $I$, then such an appearance is found
completely, although this does not ensure that a whole $H$-appearance is found. Now let $x\in P$ play the role
$v_i$. A random neighbor of $x$ leads to a subgraph $P'$ that is an appearance of some $I$-component containing the
role $v_i$, in \emph{another} $H$-appearance, in which $x$ again plays $v_i$. By~(2), for every $v_j$ with
$(v_i, v_j)\in E_D$, the $H$-appearances of $P$ and $P'$ use the same vertex in the role $v_j$. In particular, if $P$
contains this vertex $y$ and $P'$ contains the role $v_j$, then $y\in P'$. We now formally define the subgraphs
that might be observed in this way, with respect to a specific $I$ and $D$.

\begin{definition}[$(H,I,D)$-generalized cactus]\label{def:pcactus}
An $(H,I,D)$-generalized cactus is a graph $F$ obtained inductively in steps $m=1,2,\ldots,N$, for some fixed $N$. In
step $m$ we add a new subgraph $P^{(m)}$, called a \emph{petal}, which is isomorphic to some $C_i(I)$. Together
with $P^{(m)}$ we fix a list $\lambda_m=(\lambda_m(v_1),\ldots,\lambda_m(v_r))$, where $\lambda_m(v_\ell)$ is the vertex that plays
the role $v_\ell$ in the $H$-copy to which $P^{(m)}$ belongs. The list has an entry for every role in $I$, also for
roles whose vertex is not in $P^{(m)}$. Such an entry is a vertex that is not discovered yet (as $P^{(m)}$ might be only a subgraph of $H$ rather than the entire $H$), 
and it becomes a vertex of
$F$ as soon as a petal containing its role is pasted on it (see Remark~\ref{rem:lists}).
\begin{itemize}
\item[(a)] \emph{Step 1:} all entries $\lambda_1(v_1),\ldots,\lambda_1(v_r)$ are new. A petal $P^{(1)}$ isomorphic to some
$C_i(I)$ is created, and for every $v_\ell\in I$ in $C_i(I)$, its vertex playing $v_\ell$ is $\lambda_1(v_\ell)$.
\item[(b)] \emph{Step $m\ge2$:} one of the following.
\begin{itemize}
\item \emph{Fresh petal:} as in step $1$, all entries of $\lambda_m$ are new, and a petal $P^{(m)}$ isomorphic to some $C_i(I)$ is
created, whose vertex playing $v_\ell$ is $\lambda_m(v_\ell)$ for every $v_\ell\in I$ in $C_i(I)$. It is vertex disjoint from all earlier petals.
\item \emph{Attached petal:} choose $j<m$ and a role $v_{i_1}\in I$ such that $u_1:=\lambda_j(v_{i_1})$ is already a vertex of $F$, and
choose a component $C'(I)$ that contains the role $v_{i_1}$. Set
\[
\lambda_m(v_\ell):=\lambda_j(v_\ell)\quad\text{if } v_\ell=v_{i_1}\text{ or } (v_{i_1},v_\ell)\in E_D,
\]
and let every other entry of $\lambda_m$ be new. A new petal $P^{(m)}$ isomorphic to $C'(I)$ is created, and for every $v_\ell\in I$ in
$C'(I)$, its vertex playing $v_\ell$ is $\lambda_m(v_\ell)$. Thus, if $\lambda_m(v_\ell)$ is already a vertex of $F$, whether it lies in $P^{(j)}$ or in
another petal, then $P^{(m)}$ is pasted on it. Apart from these vertices, $P^{(m)}$ is vertex disjoint from all earlier petals.
\end{itemize}
\end{itemize}
The cactus is $F:=P^{(1)}\cup\cdots\cup P^{(N)}$.
\end{definition}

\begin{remark}[why every petal carries a full list]\label{rem:lists}
The dependencies in $D$ are global: if $(v_i,v_j)\in E_D$, then all $H$-appearances that use the same vertex in the role
$v_i$ also use the same vertex in the role $v_j$, even if this vertex lies in no petal seen so far. For example, let
$I=\{a,c\}$ with $(a,c)\in E_D$, and let $P^{(1)}$ be a petal that contains only the role $a$, with vertex $x$. Attach two petals
$P^{(2)}$ and $P^{(3)}$ containing the roles $a$ and $c$, both at $x$. Although $P^{(1)}$ has no vertex of role $c$, both new
petals must use the same vertex of role $c$, and they do, since both take it from the entry $\lambda_1(c)$. The same
holds when this vertex was chosen by a petal not attached at $x$, or is determined only through transitivity, as for
$(a,b),(b,c)\in E_D$. In general, all petals that use a vertex in the role $v_i$ have the same entries for all roles
determined by $v_i$, so the construction does not depend on which of these petals the next petal is attached to.
\end{remark}
\begin{figure}[ht!]
\centering
\resizebox{0.95\textwidth}{!}{%
\begin{tikzpicture}[
  hub/.style={rectangle,draw,thick,fill=gray!15,minimum size=6mm,inner sep=1pt,font=\small},
  pv/.style={circle,draw,thick,fill=white,minimum size=5.5mm,inner sep=0pt,font=\small},
  e1/.style={blue!70!black,thick},
  e2/.style={orange!85!black,thick},
  lab/.style={font=\scriptsize}]
\begin{scope}
  \node[hub] (a1) at (0,2) {$a$};
  \node[hub] (b1) at (-1.6,2) {$b$};
  \node[hub] (c1) at (-3.2,2) {$c$};
  \node[hub] (b2) at (1.6,2) {$b$};
  \node[hub] (c2) at (3.2,2) {$c$};
  \node[lab,above=1pt] at (a1.north) {$\lambda_1(a)=\lambda_2(a)$};
  \node[lab,above=1pt] at (b1.north) {$\lambda_1(b)$};
  \node[lab,above=1pt] at (c1.north) {$\lambda_1(c)$};
  \node[lab,above=1pt] at (b2.north) {$\lambda_2(b)$};
  \node[lab,above=1pt] at (c2.north) {$\lambda_2(c)$};
  \node[pv,draw=blue!70!black] (x1) at (-1.6,0) {};
  \node[pv,draw=orange!85!black] (y2) at (1.6,0) {};
  \foreach \h in {a1,b1,c1}{ \draw[e1] (x1)--(\h); }
  \foreach \h in {a1,b2,c2}{ \draw[e2] (y2)--(\h); }
  \node[blue!70!black,font=\footnotesize] at (-1.6,-0.6) {$P^{(1)}$};
  \node[orange!85!black,font=\footnotesize] at (1.6,-0.6) {$P^{(2)}$};
  \node[font=\small] at (0,-1.4) {(a) $E_D=\emptyset$};
\end{scope}
\begin{scope}[xshift=9.2cm]
  \node[hub] (b1) at (-3.1,2) {$b$};
  \node[hub] (a1) at (-1.1,2) {$a$};
  \node[hub] (c1) at (1.1,2) {$c$};
  \node[hub] (b2) at (3.1,2) {$b$};
  \node[lab,above=1pt] at (b1.north) {$\lambda_1(b)$};
  \node[lab,above=1pt] at (a1.north) {$\lambda_1(a)=\lambda_2(a)$};
  \node[lab,above=1pt] at (c1.north) {$\lambda_1(c)=\lambda_2(c)$};
  \node[lab,above=1pt] at (b2.north) {$\lambda_2(b)$};
  \node[pv,draw=blue!70!black] (x1) at (-1.6,0) {};
  \node[pv,draw=orange!85!black] (y2) at (1.6,0) {};
  \foreach \h in {b1,a1,c1}{ \draw[e1] (x1)--(\h); }
  \foreach \h in {a1,c1,b2}{ \draw[e2] (y2)--(\h); }
  \node[blue!70!black,font=\footnotesize] at (-1.6,-0.6) {$P^{(1)}$};
  \node[orange!85!black,font=\footnotesize] at (1.6,-0.6) {$P^{(2)}$};
  \node[font=\small] at (0,-1.4) {(b) $E_D=\{(a,c)\}$};
\end{scope}
\end{tikzpicture}}
\caption{Two steps of an $(H,I,D)$-generalized cactus for $H=K_{2,3}$ and $I=\{a,b,c\}$, the side of size three. The
$I$-components are stars whose center is adjacent to $a$, $b$ and $c$. Step $1$ creates the petal $P^{(1)}$. Assume that  $(a,b) \notin E_D$. Step $2$
attaches the petal $P^{(2)}$ to $P^{(1)}$ at its vertex $\lambda_1(a)$ of role $a$. \textbf{(a)}~If $D$ has no edges, only
the entry of role $a$ is copied, so $P^{(1)}$ and $P^{(2)}$ meet in one vertex. \textbf{(b)}~If $(a, c)\in E_D$, also
$\lambda_2(c)=\lambda_1(c)$. Then $P^{(1)}$ and $P^{(2)}$ meet in two vertices and form a cycle of length four. In both cases
the entry $\lambda_2(b)$ is new.}\label{fig:preorder}
\end{figure}
\begin{figure}[ht!]
\centering
\resizebox{0.9\textwidth}{!}{%
\begin{tikzpicture}[
  hub/.style={rectangle,draw,thick,fill=gray!15,minimum size=6mm,inner sep=1pt,font=\small},
  pv/.style={circle,draw,thick,fill=white,minimum size=5.5mm,inner sep=0pt,font=\scriptsize},
  lab/.style={font=\scriptsize}]
\colorlet{c1}{blue!70!black}\colorlet{c2}{orange!85!black}\colorlet{c3}{teal!70!black}\colorlet{c4}{violet!80!black}\colorlet{c5}{red!70!black}
\node[hub] (B1) at (0,0) {$b$};
\node[hub] (A1) at (2.6,1.4) {$a$};
\node[hub] (C1) at (2.6,-1.4) {$c$};
\node[hub] (B2) at (5.2,0) {$b$};
\node[hub] (A3) at (-2.6,1.4) {$a$};
\node[hub,line width=1.4pt] (C3) at (-2.6,-1.4) {$c$};
\node[hub] (B4) at (-5.2,0) {$b$};
\node[hub] (A5) at (-4.6,-3.4) {$a$};
\node[hub] (B5) at (-0.6,-3.4) {$b$};
\node[lab,below=2pt] at (B1.south) {$\lambda_1(b)=\lambda_3(b)$};
\node[lab,above=2pt] at (A1.north) {$\lambda_1(a)=\lambda_2(a)$};
\node[lab,below=2pt] at (C1.south) {$\lambda_1(c)=\lambda_2(c)$};
\node[lab,below=2pt] at (B2.south) {$\lambda_2(b)$};
\node[lab,above=2pt] at (A3.north) {$\lambda_3(a)=\lambda_4(a)$};
\node[lab,left=3pt,align=right] at (C3.west) {$\lambda_3(c)=\lambda_4(c)$\\$=\lambda_5(c)$};
\node[lab,below=2pt] at (B4.south) {$\lambda_4(b)$};
\node[lab,below=2pt] at (A5.south) {$\lambda_5(a)$};
\node[lab,below=2pt] at (B5.south) {$\lambda_5(b)$};
\node[pv,draw=c1] (x1) at (1.3,0) {};
\node[pv,draw=c2] (x2) at (3.9,0) {};
\node[pv,draw=c3] (x3) at (-1.3,0) {};
\node[pv,draw=c4] (x4) at (-3.9,0) {};
\node[pv,draw=c5] (x5) at (-2.6,-2.9) {};
\foreach \h in {B1,A1,C1}{\draw[c1,thick] (x1)--(\h);}
\foreach \h in {A1,C1,B2}{\draw[c2,thick] (x2)--(\h);}
\foreach \h in {B1,A3,C3}{\draw[c3,thick] (x3)--(\h);}
\foreach \h in {A3,C3,B4}{\draw[c4,thick] (x4)--(\h);}
\foreach \h in {C3,A5,B5}{\draw[c5,thick] (x5)--(\h);}
\node[c1,font=\footnotesize] at (1.3,0.55) {$P^{(1)}$};
\node[c2,font=\footnotesize] at (3.9,0.55) {$P^{(2)}$};
\node[c3,font=\footnotesize] at (-1.3,0.55) {$P^{(3)}$};
\node[c4,font=\footnotesize] at (-3.9,0.55) {$P^{(4)}$};
\node[c5,font=\footnotesize] at (-2.0,-2.75) {$P^{(5)}$};
\end{tikzpicture}}
\caption{An $(H,I,D)$-generalized cactus for $H=K_{2,3}$, $I=\{a,b,c\}$ and $E_D=\{(a, c)\}$. As in
Figure~\ref{fig:preorder}, each petal is a star whose center is adjacent to vertices of roles $a$, $b$ and $c$, and
each such vertex is labeled by the list entries it equals. Step $1$ creates $P^{(1)}$. Step $2$ attaches $P^{(2)}$ to
$P^{(1)}$ at its vertex of role $a$, so it copies the entries of roles $a$ and $c$, and its entry of role $b$ is new. Step
$3$ attaches $P^{(3)}$ to $P^{(1)}$ at its vertex of role $b$. Since $b$ determines no other role, only the entry of role
$b$ is copied. Step $4$ attaches $P^{(4)}$ to $P^{(3)}$ at its vertex of role $a$, so again the entries of roles $a$ and
$c$ are copied. Step $5$ attaches $P^{(5)}$ to $P^{(4)}$ at its vertex of role $c$ (thick border). Although $a$
determines $c$, the role $c$ does not determine $a$, so only the entry of role $c$ is copied, and $P^{(5)}$ gets new
vertices of roles $a$ and $b$. This shows the asymmetry of $D$: attaching at $a$ copies two entries, while attaching
at $c$ copies only one.}\label{fig:cactus_nontrivial}
\end{figure}

The grand goal of this section is to prove our main characterization result.
\begin{theorem}[main theorem]\label{thm:main}
Let $\calF$ be a finite family of connected graphs, each with at least one edge and at most $k$ vertices, and let
$p\ge k-1$. Then $\calF$-freeness is one-sided testable on $p$-degenerate graphs if and only if the following holds: for
every $H\in\calF$, every nonempty labeled independent set $I\subseteq V(H)$ and every transitively closed digraph $D$ on $I$,
some member of $\calF$ (possibly $H$ itself) is a subgraph of an $(H,I,D)$-generalized cactus.
\end{theorem}
\subsection{The lower bound}\label{sec:lb-general}

Throughout this section, let $\calF$ be a finite family of connected graphs on at most $k$ vertices, each with at
least one edge. Fix $H\in\calF$, a labeled independent set $I=\{v_1,\ldots,v_r\}\subseteq V(H)$ with $r\ge1$, and a
transitively closed directed graph $D=(I,E_D)$. Let $C_1(I),\ldots,C_t(I)$ be the $I$-components of $H$. For
$i\in[r]$ let
\[
U_i:=\{v_i\}\cup\{v_j:\ (v_i, v_j)\in E_D\}
\]
be the set of roles determined by $v_i$, including $v_i$ itself. We sometimes identify $v_i \in I$ with the index $i\in [r]$, when there is no confusion.

The following is immediate from the transitivity of $D$.
\begin{observation}\label{lem:Ui}
If $(v_i, v_j)\in E_D$, then $U_j\subseteq U_i$. If $i\ne j$ and $(v_i, v_j)\notin E_D$, then $v_j\in U_j\setminus U_i$.
\end{observation}

\paragraph{From independent roles to dependencies.} Recall the construction of Lemma~\ref{lem:lb}, and consider it for $I=\{a,b,c\}$. There are $m$ hubs $(a,1),\ldots,(a,m)$ of role $a$ (called $a_1,\ldots,a_m$ in Lemma~\ref{lem:lb}), and similarly for $b$ and $c$. For every triple $(\alpha_a,\alpha_b,\alpha_c)\in[m]^3$ there is a copy of $H$, which uses the hubs $(a,\alpha_a)$, $(b,\alpha_b)$ and $(c,\alpha_c)$. In this construction the roles are independent, that is, $E_D=\emptyset$. For example, the hub $(a,2)$ lies in the copies $(2,\beta,\gamma)$ for all $\beta,\gamma\in[m]$, and these copies use all $m$ hubs of role $c$. Suppose now that we want the dependency $(a,c)\in E_D$, that is, all copies through the same hub of role $a$ should use the same hub of role $c$. To achieve this, we split each hub of role $a$ according to the hub of role $c$ of the copy. The hub $(a,2)$ is replaced by the $m$ hubs $(a,(2,1)),\ldots,(a,(2,m))$, where $(a,(2,\gamma))$ lies in the copies $(2,\beta,\gamma)$, $\beta\in[m]$. All of these copies use the hub $(c,\gamma)$, so the hub of role $a$ now determines the hub of role $c$. Doing the same for every hub of role $a$ gives the $m^2$ hubs $(a,(\alpha_a,\alpha_c))$, while the hubs of roles $b$ and $c$ are unchanged. The hub of role $a$ must still record $\alpha_a$, and not only $\alpha_c$: the copies through $(c,3)$ use the $m$ different hubs $(a,(1,3)),\ldots,(a,(m,3))$ of role $a$, so $c$ does not determine $a$, as required by $(c,a)\notin E_D$. In general, every dependency $(v_i,v_j)\in E_D$ splits each hub of role $v_i$ according to the hub of role $v_j$. Since $D$ is transitively closed, the hub of role $v_i$ then records the coordinates of exactly the roles in $U_i$.

The following definition formally generalizes the construction of Lemma~\ref{lem:lb} (in which $E_D=\emptyset$) to an arbitrary
transitively closed digraph $D=(I,E_D)$.
\begin{definition}\label{def:LB_ilan}
Let $m\ge2$, and let $\circ$ be a new symbol. An \emph{index} is a vector
$\alpha=(\alpha_1,\ldots,\alpha_r,\alpha_\circ)\in[m]^{r+1}$, and for $U\subseteq[r]$ we write $\alpha|_U$ for its restriction to the
coordinates in $U$. The graph $G=G_m(H,I,D)$ has the following vertices:
\begin{itemize}
\item \emph{hubs} $(i,\beta)$, for $i\in[r]$ and $\beta\in[m]^{U_i}$,
\item \emph{private vertices} $(\alpha,z)$, for $\alpha\in[m]^{r+1}$ and $z\in V(H)\setminus I$.
\end{itemize}
The edges of $G$ are of two kinds. For every index $\alpha$:
\begin{itemize}
\item $(\alpha,z)$ is adjacent to $(\alpha,z')$ whenever $(z,z')\in E(H)$.
\item $(\alpha,z)$ is adjacent to the hub $(i,\alpha|_{U_i})$ whenever $(z,v_i)\in E(H)$.
\end{itemize}
There are no edges between hubs. Let $n$ be the number of vertices of $G$. The distribution $\calD$ relabels $G$ with
a uniformly random bijection $V(G)\to[n]$. See Figure~\ref{fig:example_hubs} for an illustration.
\end{definition}

In other words, for each tuple $\alpha$, the private vertices $(\alpha,z)$ together with the hubs $\lambda_\alpha(v_i):=(i,\alpha|_{U_i})$, $i\in[r]$,
span a copy of $H$, the \emph{copy $\alpha$}, in which $z\notin I$ is played by $(\alpha,z)$ and $v_i$ by $\lambda_\alpha(v_i)$. We call
$\lambda_\alpha=(\lambda_\alpha(v_1),\ldots,\lambda_\alpha(v_r))$ the \emph{list} of the copy $\alpha$, in analogy with the lists of the petals in
Definition~\ref{def:pcactus}.  For $s\in[t]$, the subgraph $P_{\alpha,s}:=\phi_\alpha(C_s(I))$ of $G$ is the \emph{petal} of the copy $\alpha$ corresponding to $C_s(I)$, where $\phi_\alpha:V(H)\to V(G)$ maps $z\notin I$ to $(\alpha,z)$ and $v_i$ to $\lambda_\alpha(v_i)$. Every private vertex lies in exactly one petal.

In words, the hub of role $v_i$ remembers the coordinates of all roles determined by $v_i$. If $(v_i,v_j)\in E_D$,
then $U_j\subseteq U_i$ by Observation~\ref{lem:Ui}, so the hub of role $v_i$ determines the hub of role $v_j$. If $i\ne j$ and
$(v_i,v_j)\notin E_D$, then $j\notin U_i$, so the coordinate $\alpha_j$ is free among the copies through a hub of role $v_i$,
and these copies use many different hubs of role $v_j$. For $E_D=\emptyset$ we have $U_i=\{i\}$, so the hub of role $v_i$ is
$(i,\alpha_i)$, and we obtain the construction of Lemma~\ref{lem:lb}, with the additional  $\circ$-coordinate.

The coordinate $\circ$ is ``remembered'' by no hub, so copies whose indices differ only in $\alpha_\circ$ have the same lists
and different private vertices. Its only purpose is technical -- to ensure, as explained,  that \emph{every} hub lies in at least $m$ copies. This matters when some role $v_i$ determines all roles, that is, $U_i=[r]$. Without $\circ$, the hub of role $v_i$ would
record the whole index of its copy and lie in a single copy, so a neighbor query at it would always return a vertex
of this copy. This does not help the algorithm: a petal attached at such a hub copies the entire list, as another petal
of the same copy does, so the explored part is still a generalized cactus. It would, however, complicate the proof.

Below, we expand the example with $I=\{a,b,c\}$ started above.
\paragraph{Example.} Let $I=\{a,b,c\}$, $m=3$ and $E_D=\{(a,c)\}$, so that $U_a=\{a,c\}$, $U_b=\{b\}$ and $U_c=\{c\}$.
 An index is a vector
$\alpha=(\alpha_a,\alpha_b,\alpha_c,\alpha_\circ)\in[3]^4$. The copy $\alpha=(2,1,3,2)$ has the list
\[
\lambda_\alpha(a)=\big(a,\alpha|_{U_a}\big)=\big(a,(2,3)\big),\qquad \lambda_\alpha(b)=\big(b,(1)\big),\qquad \lambda_\alpha(c)=\big(c,(3)\big),
\]
where $(2,3)$ lists the coordinates $\alpha_a$ and $\alpha_c$ (see Figure~\ref{fig:example_hubs}).
\begin{itemize}
\item None of these entries depends on $\alpha_\circ$, so the copies $(2,1,3,1)$, $(2,1,3,2)$ and $(2,1,3,3)$ have the
same list, and differ only in their private vertices.
\item Every copy $\alpha'$ with $\lambda_{\alpha'}(a)=\lambda_\alpha(a)$ has $\alpha'_a=2$ and $\alpha'_c=3$, and hence $\lambda_{\alpha'}(c)=\lambda_\alpha(c)$, as
required by $(a,c)\in E_D$.
\item The copy $\alpha'=(2,2,3,1)$ has the same entries of roles $a$ and $c$ as $\alpha$, but
$\lambda_{\alpha'}(b)=\big(b,(2)\big)\ne\lambda_\alpha(b)$, since $(a,b)\notin E_D$.
\item The copy $\alpha'=(1,1,3,1)$ has the same entry of role $c$ as $\alpha$, but $\lambda_{\alpha'}(a)=\big(a,(1,3)\big)\ne\lambda_\alpha(a)$,
since $(c,a)\notin E_D$.
\end{itemize}

\begin{figure}[ht]
	\centering
	\resizebox{\textwidth}{!}{%
	\begin{tikzpicture}[
		hub/.style={rectangle,draw,thick,fill=gray!15,minimum height=7mm,inner sep=3pt,font=\small},
		cp/.style={rectangle,rounded corners=3pt,draw,thick,fill=white,minimum height=7mm,inner sep=3pt,font=\small},
		lab/.style={font=\scriptsize}]
		\colorlet{cB}{blue!70!black}\colorlet{cO}{orange!85!black}\colorlet{cG}{teal!70!black}
		\node[hub] (a2) at (-8.2,3) {$\big(a,(1,3)\big)$};
		\node[hub] (b1) at (-4.8,3) {$\big(b,(1)\big)$};
		\node[hub,line width=1.4pt] (a)  at (-1.4,3) {$\big(a,(2,3)\big)$};
		\node[hub,line width=1.4pt] (c)  at (1.6,3)  {$\big(c,(3)\big)$};
		\node[hub] (b2) at (4.8,3)  {$\big(b,(2)\big)$};
		\node[lab,above=1pt] at (a2.north) {role $a$};
		\node[lab,above=1pt] at (b1.north) {role $b$};
		\node[lab,above=1pt] at (a.north)  {role $a$};
		\node[lab,above=1pt] at (c.north)  {role $c$};
		\node[lab,above=1pt] at (b2.north) {role $b$};
		\node[cp,draw=cG] (k0) at (-8.2,0) {$(1,1,3,1)$};
		\node[cp,draw=cB] (k1) at (-4.8,0) {$(2,1,3,1)$};
		\node[cp,draw=cB] (k2) at (-2.1,0) {$(2,1,3,2)$};
		\node[cp,draw=cB] (k3) at (0.6,0)  {$(2,1,3,3)$};
		\node[cp,draw=cO] (k4) at (4.2,0)  {$(2,2,3,1)$};
		\foreach \k in {k1,k2,k3}{
			\draw[cB,thick] (\k.north)--(b1.south);
			\draw[cB,thick] (\k.north)--(a.south);
			\draw[cB,thick] (\k.north)--(c.south);}
		\draw[cO,thick] (k4.north)--(a.south);
		\draw[cO,thick] (k4.north)--(c.south);
		\draw[cO,thick] (k4.north)--(b2.south);
		\draw[cG,thick] (k0.north)--(a2.south);
		\draw[cG,thick] (k0.north)--(b1.south);
		\draw[cG,thick] (k0.north) to[out=20,in=-150] (c.south);
		\node[lab,text=cB] at (-2.1,-0.75) {same hubs, differ only in $\alpha_\circ$};
		\node[lab,text=cO,align=center] at (4.2,-0.85) {same $a$- and $c$-hubs,\\ different $b$-hub};
		\node[lab,text=cG,align=center] at (-8.2,-0.85) {same $c$-hub,\\ different $a$-hub};
	\end{tikzpicture}}
	\caption{The construction for $I=\{a,b,c\}$, $m=3$ and $E_D=\{(a,c)\}$, so that $U_a=\{a,c\}$, $U_b=\{b\}$ and
		$U_c=\{c\}$. Each rounded box in the bottom of the Figure is a copy, labelled by its $[m]^4$-tuple $(\alpha_a,\alpha_b,\alpha_c,\alpha_\circ)$, and stands for its
		private vertices. 
        It is joined to its three hubs, the entries $\lambda_\alpha(s)=(s,\alpha|_{U_s})$ of its list. The copies
		$(2,1,3,1)$, $(2,1,3,2)$ and $(2,1,3,3)$ differ only in $\alpha_\circ$ and have the same hubs.
        The copy $(2,2,3,1)$ has
		the same $a$-role hub as the copy $(2,1,3,1)$, hence also the same hub of role $c$, as $(a,c)\in E_D$, but a different hub of role $b$.
		The copy $(1,1,3,1)$ has the same hub of role $c$, but a different hub of role $a$, as $(c,a)\notin E_D$.}\label{fig:example_hubs}
\end{figure}

Next, we prove the following properties for the above construction.
\begin{lemma}[properties of the construction]\label{lem:LB_props}
Let $d:=r+1$.
\begin{enumerate}
\item[(a)] Every $\alpha \in [m]^d$ forms  a copy of $H$, and the $m^d$ copies are pairwise edge-disjoint.
\item[(b)] If $(v_i,v_j)\in E_D$ and $\lambda_\alpha(v_i)=\lambda_{\alpha'}(v_i)$ for two copies $\alpha,\alpha'$, then
$\lambda_\alpha(v_j)=\lambda_{\alpha'}(v_j)$.
\item[(c)] A hub of role $v_i$ lies in exactly $m^{d-|U_i|}$ copies, which is between $m$ and $m^{d-1}$.
\item[(d)] If $i\ne j$ and $(v_i,v_j)\notin E_D$, then a hub of role $v_i$ and a hub of role $v_j$ lie together in at
most a $\frac1m$ fraction of the copies through the hub of role $v_i$.
\item[(e)] All neighbors of a private vertex lie in its petal, and it has degree at most $k-1$. The hubs form an
independent set, so $G$ is $(k-1)$-degenerate. There are at most $rm^{d-1}$ hubs.
\item[(f)] $G$ has no isolated vertices, $m\ge(n/k)^{1/d}$, and $G$ is $\frac1k$-far from $\calF$-freeness.
\end{enumerate}
\end{lemma}
\begin{proof}
(a) The hubs of a copy have distinct roles and its private vertices belong to no other copy. Every edge of $H$ joins two vertices outside $I$, or a vertex outside $I$ and a vertex of $I$, so it is mapped to an edge of $G$. Every edge of
the copy has a private endpoint, so distinct copies are edge-disjoint.

(b) The equality $\lambda_\alpha(v_i)=\lambda_{\alpha'}(v_i)$ means that $\alpha|_{U_i}=\alpha'|_{U_i}$. By Observation~\ref{lem:Ui}, $U_j\subseteq U_i$, so
$\alpha|_{U_j}=\alpha'|_{U_j}$, that is, $\lambda_\alpha(v_j)=\lambda_{\alpha'}(v_j)$.

(c) The copies through $(i,\beta)$ are those with $\alpha|_{U_i}=\beta$, and $1\le|U_i|\le r=d-1$.

(d) A copy through $(i,\beta)$ and $(j,\beta')$ has all its coordinates in $U_i\cup U_j$ fixed. By Observation~\ref{lem:Ui}, $j\in U_j\setminus U_i$, so at least $|U_i|+1$ coordinates are fixed, and there are at most $m^{d-|U_i|-1}$ such copies.

(e) The first statement follows by construction, and the number of hubs is $\sum_i m^{|U_i|}\le rm^{d-1}$.

(f) Every vertex lies in a copy of $H$, which has no isolated vertices. By~(e) there are at most $rm^{d-1}$ hubs. There are exactly $(k-r)m^d$ private vertices,
one for each of the $m^d$ indices $\alpha$ and each of the $k-r$ vertices of $H$ outside $I$. Hence we have $n\le rm^{d-1}+(k-r)m^d\le km^d$. By
(a), destroying all copies requires deleting at least $m^d\ge n/k$ edges, and $H\in\calF$.\end{proof}

\paragraph{Setting.} Fix a canonical algorithm that makes $q=O(1)$ queries to $G_m$. Each query is a \emph{vertex sample},
returning a uniformly random vertex, or a \emph{random neighbor query} at an already seen vertex.
The algorithm \emph{enters} a petal at the first query that returns one of its private vertices. Let $\bT$ be
the set of copies with an entered petal, and let $\bH_{\bT}:=\{\lambda_\alpha(v_i):\ \alpha\in\bT,\ i\in[r]\}$ be the set of all
entries of the corresponding lists, including hubs that have not been seen. Each query enters at most one petal, so
$|\bT|\le q$ and $|\bH_{\bT}|\le rq$. Let $\calbB_1,\calbB_2,\calbB_3$ be the events that at some query the
following happens, respectively, and let $\calbB:=\calbB_1\cup \calbB_2\cup \calbB_3$. We call these events \emph{bad}, since each of them is a way in which the explored part may leave the class of
generalized cacti, and thus a way in which the algorithm may make an 'unexpected' progress towards finding a member of $\calF$.
\begin{itemize}
\item[$\calbB_1$:] A new petal of a copy that is already in $\bT$ is entered. This is bad because two petals of the same copy share all their common hubs, and together the petals of one copy form $H$.
\item[$\calbB_2$:] A neighbor query at a hub $h$ of role $v_i$ enters a petal of a copy $\alpha$ such that $\lambda_\alpha(v_j)\in\bH_{\bT}$ for some $j\ne i$ with $(v_i,v_j)\notin E_D$. The new petal is then glued to the explored part at a vertex that $v_i$ does not determine, which no attached step allows, and this may close a cycle through two hubs.

\item[$\calbB_3$:] A vertex sample returns a hub, or enters a petal of a copy $\alpha$ such that some entry of $\lambda_\alpha$ lies in $\bH_{\bT}$.  A petal found by a vertex sample should be a new one, and a hub found by a vertex sample could serve as an attachment point without being a vertex of an entered petal.
\end{itemize}

The next lemma shows that, as long as none of these events occurs, the entered petals form generalized cacti. So
these three events are the only ways in which this can fail.

\begin{lemma}\label{lem:structure}
If $\calbB$ does not occur, then the discovered graph is contained in an $(H,I,D)$-generalized cactus.
\end{lemma}
\begin{proof}
Let $\bP_1,\bP_2,\ldots$ be the entered petals in order of entry, and let $\balpha_m$ be the copy of $\bP_m$. By $\calbB_1$, the copies $\balpha_m$ are pairwise distinct, so their private vertices are disjoint, although they may share hubs. We give $\bP_m$ the list $\lambda_m:=\lambda_{\balpha_m}$ of its copy. We call $\bP_m$ \emph{fresh} if all entries of $\lambda_m$ are new, and \emph{attached} if it arises as an attached petal in
Definition~\ref{def:pcactus}. We show that every $\bP_m$ is fresh or attached. At the moment $\bP_m$ is entered, the entries of the earlier lists are exactly the hubs in $\bH_{\bT}$.

Suppose first that $\bP_m$ was entered by a vertex sample. By $\calbB_3$, no entry of $\lambda_m$ lies in $\bH_{\bT}$, so all entries of $\lambda_m$ are new, and $\bP_m$ shares no vertex with the earlier petals.

Suppose now that $\bP_m$ was entered by a neighbor query at a vertex $u_1$. A neighbor query at a private vertex returns a vertex of its own petal, which has already been entered, so $u_1$ is a hub, say of role $v_{i_1}$. By $\calbB_3$, no vertex sample has returned a hub, so $u_1$ was first seen as the neighbor of a private vertex of an earlier entered petal $\bP_j$. Hence $u_1=\lambda_j(v_{i_1})$ is a vertex of $\bP_j$, and $u_1=\lambda_m(v_{i_1})$ is a vertex of $\bP_m$. For $v_\ell$ with $(v_{i_1},v_\ell)\in E_D$, Lemma~\ref{lem:LB_props}(b) gives $\lambda_m(v_\ell)=\lambda_j(v_\ell)$. For all other roles $v_\ell\ne v_{i_1}$, the entry $\lambda_m(v_\ell)$ is not in $\bH_{\bT}$ by $\calbB_2$, so it is new. So $\bP_m$ is an attached petal.

So every entered petal is a fresh or an attached petal as in Definition~\ref{def:pcactus}, and listing the entered petals in their order of entry gives a construction of a generalized cactus. Finally, every discovered edge has a private endpoint $w$. The algorithm has seen $w$, so the petal containing $w$ has been entered, and all neighbors of $w$ lie in this petal. So every discovered edge lies in an entered petal.
\end{proof}

\begin{lemma}\label{lem:probability}
$\Pr[\calbB]\le(2r+1)q^2/m$.
\end{lemma}
\begin{proof}
 We bound the probability of $\calbB$ at a single query, conditioned on the history. A neighbor query at a private
vertex enters no new petal. Consider a neighbor query at a hub $h$ of role $v_i$. Each copy through $h$ contributes
the same number $\deg_H(v_i)$ of neighbors of $h$, so the copy of the answer is uniform among the copies through $h$, of which there are at least $m$ by Lemma~\ref{lem:LB_props}(c). So $\calbB_1$ has probability at most
$|\bT|/m\le q/m$. By Lemma~\ref{lem:LB_props}(d), each of the at most $rq$ hubs of $\bH_{\bT}$ of a role $v_j$ with
$j\ne i$ and $(v_i,v_j)\notin E_D$ lies in the copy of the answer with probability at most $1/m$, so $\calbB_2$ has
probability at most $rq/m$. Consider now a vertex sample. It returns a hub with probability at most
$rm^{d-1}/((k-r)m^d)\le r/m$. At most $|\bT|+|\bH_{\bT}|m^{d-1}$ copies lie in $\bT$ or have a list entry in
$\bH_{\bT}$, and each has $k-r$ private vertices, so the sample enters a petal of one of them with probability at most
$(q+rqm^{d-1})/m^d\le(r+1)q/m$. Hence each query triggers $\calbB$ with probability at most $((r+1)q+r)/m$, and a
union bound over the $q$ queries gives the claim.
\end{proof}

\begin{lemma}[general lower bound]\label{thm:LB_gen}
Suppose that no graph $H'\in\calF$ is a subgraph of any $(H,I,D)$-generalized cactus. Then every one-sided $\eps$-tester for $\calF$-freeness on $p$-degenerate graphs, with $p\ge k-1$ and $\eps\le\frac1k$, makes $\Omega\big((n/k)^{1/(2(r+1))}\big)$ queries. In particular, $\calF$-freeness is not testable.
\end{lemma}
\begin{proof}
Every graph in the support of $\calD$ is $(k-1)$-degenerate and $\frac1k$-far from $\calF$-freeness, by Lemma~\ref{lem:LB_props}(e). A one-sided tester rejects only after discovering a copy of some $H'\in\calF$. By Lemma~\ref{lem:structure} and the hypothesis, this is impossible unless $\calbB$ occurs. By Lemma~\ref{lem:probability}, the tester rejects with probability at most $(2r+1)q^2/m$. Rejecting with probability at least $2/3$ forces $q=\Omega(\sqrt m)$, and $m\ge(n/k)^{1/(r+1)}$ by Lemma~\ref{lem:LB_props}(f).
\end{proof}

\subsection{The upper bound}\label{sec:family-ub}
 
\paragraph{Overview.} Let $G$ be $p$-degenerate and far from $\calF$-free. We fix a suitable parameter $h$ (threshold for heavy), and define a heavy vertex in $G$ as a vertex of degree significantly larger than $h$. 

We first show that we may assume that $G$ contains many edge-disjoint copies of one member $H \in \calF$, each with at least one heavy vertex, and a labeled independent set $I$ of $H$ in which all $H$-role vertices are heavy in $G$ and ``role-consistent" -- i.e., every $v$ has the same $I$-role in all the $H$-appearances it is a part of. We then show that we may assume that every pair of roles is either \emph{locked}, meaning that the hub of the first role determines the hub of the second, or \emph{spread}, meaning that it does not. The locked pairs form a transitively closed digraph $D$ on $I$. All this is done by ``cleaning" of the $H$-appearances in $G$, so to result in the final wanted $G'$ with the above properties. We will show that this can be done by deleting a relatively small number of edges leaving $G'$ far from $\calH$-freeness.

We only use the {\em existence} of this subgraph that we now refer to as $G$, the same as remarked in Section~\ref{sec:technical_overview}. Namely, we operate under the assumption that the input graph is this final $G$, and claim that the behaviour of the tester will not change significantly, as the number of removed edges is small.  

By our assumption on $\calF$, some member of $\calF$ lies in an $(H,I,D)$-generalized cactus. We show that an idealized random process, which follows the construction of this cactus inside $G$, finds that member with good probability, and that a bounded-depth random BFS realizes this process.

\paragraph{Assumptions.} Throughout this section, $\calF$ is a finite family of connected graphs, each with at least one edge and at most $k$ vertices, and $p\ge1$, $\eps>0$. We assume that for every $H\in\calF$, every nonempty labeled independent set $I\subseteq V(H)$ and every transitively closed digraph $D$ on $I$, some member of $\calF$ is a subgraph of an $(H,I,D)$-generalized cactus. This also holds for $I=\emptyset$, since then the only  $I$-component of $H$ is $H$ itself, and a single petal is a copy of $H$. For each such triple $(H,I,D)$, fix a construction as in Definition~\ref{def:pcactus} with the least number of steps whose graph contains a member of $\calF$, and let $N$ be the largest number of steps among these constructions. A fresh petal starts a new connected component of the cactus, and later petals are pasted only on vertices of the component they are attached to, so components never merge. Since the members of $\calF$ are connected, a member contained in the graph of a construction lies in a single component, and deleting the steps of all other components gives a valid construction with fewer steps that still contains it. So each fixed construction has a single fresh petal, namely $P^{(1)}$, and all its later steps are attached. All other parameters are introduced in the lemmas where they are needed, and they are fixed in the proof of Lemma~\ref{lem:general_UB}, where we also explain why $N$ depends only on $\calF$.

\paragraph{Role consistency.} In the copies we will work with, the heavy vertices play the roles in a fixed independent set $I$. A heavy vertex may, however, lie in many copies, and a priori it could play different roles in different copies. The following lemma removes this ambiguity at the cost of a constant fraction of the copies. Afterwards every heavy vertex has a single role, which is what allows us to speak of \emph{the} hub of role $v_i$ of a copy.
 
\begin{lemma}[role consistency]\label{lem:role-cons}
Let $H$ and $G$ be graphs, and let $\calK$ be a family of copies of $H$ in $G$, each with a fixed isomorphism
$\psi_K:V(H)\to V(K)$. Let $X\subseteq V(G)$, and suppose there is a set $I=\{v_1,\ldots,v_r\}\subseteq V(H)$ with
$\psi_K^{-1}(X\cap V(K))=I$ for every $K\in\calK$, that is, in every copy the vertices lying in $X$ play exactly the roles
in $I$. Then there are a subfamily $\calK'\subseteq\calK$ with $|\calK'|\ge|\calK|/r^r$ and a map $\rho:X\to I$ such that
$\psi_K^{-1}(u)=\rho(u)$ for every $K\in\calK'$ and every $u\in X\cap V(K)$.
\end{lemma}

\begin{proof}
If $r=0$, there is nothing to prove, so let $r\ge1$. Let $\bchi:V(G)\to[r]$ be a random coloring, in which every vertex receives an independent uniform color. A copy $K$ is \emph{well colored} under a coloring $\chi$ if the colors $\chi(\psi_K(v_1)),\ldots,\chi(\psi_K(v_r))$ are distinct. Since the vertices $\psi_K(v_1),\ldots,\psi_K(v_r)$ are distinct, each copy is well colored under $\bchi$ with probability $r!/r^r$. Hence the expected number of well-colored copies is $|\calK|\,r!/r^r$, and we can fix a coloring $\chi^*$ under which at least $|\calK|\,r!/r^r$ copies are well colored.

Each well-colored copy $K$ determines a bijection $\sigma_K:[r]\to I$, given by $\sigma_K(\chi^*(\psi_K(v_i)))=v_i$ for $i\in[r]$. There are $r!$ bijections, so some bijection $\sigma$ is shared by at least a $1/r!$ fraction of the well-colored copies. Let $\calK'$ be the set of these copies. Then $|\calK'|\ge|\calK|/r^r$.
 
Set $\rho:=\sigma\circ\chi^*$ on $X$. Let $K\in\calK'$ and $u\in X\cap V(K)$, and write $v_i:=\psi_K^{-1}(u)$, which lies in $I$. Then $\rho(u)=\sigma(\chi^*(u))=\sigma_K(\chi^*(\psi_K(v_i)))=v_i$.
\end{proof}

\paragraph{A good-collection with a degree gap.} A walk from a high-degree vertex $u$ is useful only if it reaches a new copy, that is, one that avoids the low-degree vertices seen before. A single degree threshold $h$ does not guarantee this. For example, a vertex $w$ of degree $h-1$ and a vertex $u$ of degree $h$ may lie in exactly the same copies. Then every walk from $u$ leads to a copy that contains $w$, so the tester never reaches a new copy. Since $w$ lies in at most $\deg(w)$ copies, this cannot happen if $u$ lies in many more copies than $\deg(w)$. We therefore want two thresholds $h\ll h'$ such that every vertex of our copies has degree less than $h$, and is called \emph{light}, or at least $h'$, and is called \emph{heavy}. To find them, we split the degrees into $2k$ bands $[h_j,h_{j+1})$ with large ratios $h_{j+1}/h_j$. Each copy meets at most $k$ of them, so some band $[h_b,h_{b+1})$ is avoided by at least half of the copies, and we take $h:=h_b$ and $h':=h_{b+1}$.

\begin{lemma}[a good-collection]\label{lem:collection}
Let $h_1\le h_2\le\dots\le h_{2k+1}$ with $h_1\ge4p^2/\eps$, and let $c_0:=\eps/(4k^2|\calF|\,2^kk^k)$. Let $G$ be a $p$-degenerate graph on $n$ vertices that is $\eps$-far from $\calF$-freeness. Then there are a member $H\in\calF$, an index $b\in[2k]$, a labeled independent set $I=\{v_1,\ldots,v_r\}\subseteq V(H)$, and a family $\calK_1$ of at least $c_0n$ pairwise edge-disjoint copies of $H$ in $G$, each copy $K$ given with an isomorphism $\psi_K:V(H)\to V(K)$, with the following properties. Here, in the notation of Definition~\ref{def:semi-bipartite}, the vertices in $\Light_{h_b}$ are called light and the vertices in $\Heavy_{h_{b+1}}$ are called heavy. Thus, light vertices have degree less than $h_b$ and heavy vertices have degree at least $h_{b+1}$ in $G$.
\begin{enumerate}
\item[(a)] Every vertex of a copy $K\in\calK_1$ is light or heavy, and the heavy vertices of $K$ are exactly the vertices of $\psi_K(I)$.
\item[(b)] No edge of a copy in $\calK_1$ joins two heavy vertices.
\item[(c)] Every heavy vertex plays the same role in all copies containing it: there is a map $\rho$ from the heavy vertices to $I$ with $\psi_K^{-1}(u)=\rho(u)$ for every $K\in\calK_1$ and every heavy $u\in V(K)$.
\item[(d)] Every light vertex lies in fewer than $h_b$ copies of $\calK_1$.
\item[(e)] For every $\beta>0$ and every subfamily $\calK'\subseteq\calK_1$, at most $2\beta pn$ copies of $\calK'$ contain a heavy vertex $u$ with $\deg_{\calK'}(u)<\beta\deg(u)$, where $\deg_{\calK'}(u):=|\{K\in\calK':u\in V(K)\}|$.
\end{enumerate}
\end{lemma}
\begin{proof} First, let $G'$ be obtained from $G$ by deleting all edges inside $\Heavy_{h_1}$. As in the proof of Lemma~\ref{cl:bipar}, we have $|\Heavy_{h_1}|\le2pn/h_1$, so $G[\Heavy_{h_1}]$ has at most $2p^2n/h_1\le\eps n/2$ edges. Hence $G'$ is $\eps/2$-far from $\calF$-freeness, and $\Heavy_{h_1}$ is independent in $G'$. Take a maximal family of pairwise edge-disjoint copies of members of $\calF$ in $G'$. Deleting all their edges leaves an $\calF$-free graph, and each copy has fewer than $k^2$ edges, so the family has at least $\eps n/(2k^2)$ copies (as in Lemma~\ref{lemma:edge-disjoint-app}). Hence some $H\in\calF$ has at least $\eps n/(2k^2|\calF|)$ of these copies. We keep only these, together with isomorphisms $\psi_K:V(H)\to V(K)$.

Next we find the gap. For $j\in[2k]$, call the degrees in $[h_j,h_{j+1})$ the $j$-th band. Each copy has at most $k$ vertices, so it contains vertices of degree in at most $k$ of the $2k$ bands. Hence some band $b\in[2k]$ is met by at most half of the copies. We discard these copies. Every vertex of a remaining copy has degree less than $h_b$ or at least $h_{b+1}$, that is, it is light or heavy. Heavy vertices lie in $\Heavy_{h_1}$, which is independent in $G'$, so no edge of a copy joins two heavy vertices. This gives (b) and the first part of (a).

Now for each remaining copy $K$, let $I_K:=\psi_K^{-1}\big(V(K)\cap\Heavy_{h_{b+1}}\big)\subseteq V(H)$ be the set of roles played by heavy vertices. There are at most $2^k$ possible sets, so some set $I$ equals $I_K$ for at least a $2^{-k}$ fraction of the copies. We keep these copies, and we label the vertices of $I$ as $v_1,\ldots,v_r$. This gives the second part of (a), and $I$ is independent in $H$ by (b). So far we have kept at least $\eps n/(4k^2|\calF|\,2^k)$ copies. Applying Lemma~\ref{lem:role-cons} with $X$ the set of heavy vertices keeps at least a $r^{-r}\ge k^{-k}$ fraction of them and gives the family $\calK_1$ and the map $\rho$ in (c). Properties (a) and (b) are inherited by $\calK_1$.

For (d), let $w$ be light. Every copy containing $w$ uses an edge at $w$, since members of $\calF$ have no isolated vertices, and the copies are edge-disjoint. So $w$ lies in at most $\deg(w)<h_b$ copies.

Finally, for (e) let $U$ be the set of heavy vertices $u$ with $\deg_{\calK'}(u)<\beta\deg(u)$. By definition, the copies we have to count are exactly the copies of $\calK'$ that contain a vertex of $U$. Each of them is one of the $\deg_{\calK'}(u)$ copies containing some $u\in U$, so their number is at most
\[
\sum_{u\in U}\deg_{\calK'}(u)<\beta\sum_{u\in U}\deg(u)\le\beta\sum_{u\in V(G)}\deg(u)=2\beta|E(G)|\le2\beta pn,
\]
using that a $p$-degenerate graph on $n$ vertices has at most $pn$ edges.
\end{proof}

From now on we fix $H$, $b$, $I$, $\calK_1$ and $\rho$ as in Lemma~\ref{lem:collection}. The heavy vertices of a copy are its hubs, one for each role in $I$, and all its other vertices are light. Light vertices have small degree, so a random BFS sees all of their neighbors, and they lie in few copies by (d). A heavy vertex lies in many copies, and a walk from it enters a random one of them. To follow a cactus, we must know which hubs such a random copy shares with the copies seen before. This is what the next step arranges.

\paragraph{Locked and spread pairs.} In the situation of Lemma~\ref{lem:collection}, write $K(v_i):=\psi_K(v_i)$ for the hub of role $v_i$ of a copy $K$. For a family $\calK'\subseteq\calK_1$, a heavy vertex $u$ and a vertex $w$, let $c_{\calK'}(u,w):=|\{K\in\calK':u,w\in V(K)\}|$. For $i\ne j$, the pair $(v_i,v_j)$ is \emph{locked} in $\calK'$ if every heavy $u$ with $\rho(u)=v_i$ that lies in $\calK'$ has the same hub of role $v_j$ in all copies of $\calK'$ containing it. This property passes to subfamilies. For $\tau>0$, the pair $(v_i,v_j)$ is \emph{$\tau$-spread} at a heavy $u$ with $\rho(u)=v_i$ if $c_{\calK'}(u,w)\le\tau\deg_{\calK'}(u)$ for every vertex $w$ of role $v_j$, that is, no single vertex is the hub of role $v_j$ of more than a $\tau$ fraction of the copies of $\calK'$ containing $u$.

\begin{definition}[good and bad hubs]\label{def:good}
Let $\calK'\subseteq\calK_1$, let $D=(I,E_D)$ be a digraph, and let $\beta,\tau>0$. A heavy vertex $u$ with $\rho(u)=v_i$ is \emph{good} for
$(\calK',D)$ if it lies in many copies of $\calK'$, namely $\deg_{\calK'}(u)\ge\beta\deg(u)$, and if $(v_i,v_j)$ is $\tau$-spread at $u$ in $\calK'$ for
every $j\ne i$ with $(v_i,v_j)\notin E_D$. Otherwise $u$ is \emph{bad} for $(\calK',D)$.
\end{definition}

The next lemma prunes the copies until, for every pair of roles, one of these two extremes holds at almost every hub. The locked pairs then form the digraph $D$ of a generalized cactus.
\begin{lemma}\label{lem:reg}
Let $\tau,\delta\in(0,1]$ and $\beta>0$, and let $c_1:=c_0(\tau\delta)^{k^2}$ and $\eta:=2\beta p/c_1+k^2\delta$. In the setting of Lemma~\ref{lem:collection}, there are a family $\calK\subseteq\calK_1$ with $|\calK|\ge c_1n$ and a transitively closed digraph $D=(I,E_D)$ with the following properties.
\begin{enumerate}
\item[(a)] For $i\ne j$, we have $(v_i,v_j)\in E_D$ if and only if $(v_i,v_j)$ is locked in $\calK$, that is, any two copies of $\calK$ with the same hub of role $v_i$ have the same hub of role $v_j$.
\item[(b)] At most $\eta|\calK|$ copies of $\calK$ contain a vertex that is bad for $(\calK,D)$.
\end{enumerate}
\end{lemma}

\begin{proof} If $I=\emptyset$, take $\calK:=\calK_1$ and $E_D:=\emptyset$. Then (a) is vacuous and no vertex is heavy. Otherwise, start from $\calK:=\calK_1$ and repeat the following. Suppose some pair $(v_i,v_j)$ with $i\ne j$ is not locked, and the set $U$ of heavy $u$ with $\rho(u)=v_i$ at which $(v_i,v_j)$ is not $\tau$-spread satisfies $\sum_{u\in U}\deg_\calK(u)>\delta|\calK|$. For each $u\in U$, fix a vertex $w_u$ that is the hub of role $v_j$ of more than $\tau\deg_\calK(u)$ copies containing $u$, and replace $\calK$ by $\{K\in\calK:\ K(v_i)\in U,\ K(v_j)=w_{K(v_i)}\}$. This keeps at least a $\tau\delta$ fraction of the copies and makes $(v_i,v_j)$ locked. Locked pairs are never lost, so there are at most $r(r-1)\le k^2$ rounds, and the final family has at least $c_0(\tau\delta)^{k^2}n=c_1n$ copies.
 
Let $E_D$ be the set of pairs $(v_i,v_j)$ with $i\ne j$ that are locked in the final $\calK$. Then (a) holds by definition. The digraph $D$ is transitively closed. Indeed, let $(v_i,v_j)$ and $(v_j,v_\ell)$ be locked with $i\ne\ell$, and let $x$ be a heavy vertex with $\rho(x)=v_i$. Then $x$ has a unique hub $y$ of role $v_j$ in its copies, all copies containing $x$ contain $y$, and in them $y$ has a unique hub of role $v_\ell$. So $(v_i,v_\ell)$ is locked.

For (b), by Lemma~\ref{lem:collection}(e) the copies containing a vertex $u$ with $\deg_\calK(u)<\beta\deg(u)$ number at most $2\beta pn\le(2\beta p/c_1)|\calK|$. Since the rule above no longer applies, for each of the at most $k^2$ pairs $(v_i,v_j)$ with $i\ne j$ and $(v_i,v_j)\notin E_D$, at most $\delta|\calK|$ copies have a hub $u$ of role $v_i$ at which $(v_i,v_j)$ is not $\tau$-spread. In total, at most $(2\beta p/c_1+k^2\delta)|\calK|=\eta|\calK|$ copies contain a bad vertex.
\end{proof}

From now on, \emph{good} and \emph{bad} refer to the family $\calK$ and the digraph $D$ of Lemma~\ref{lem:reg}. We now have a triple $(H,I,D)$, where $D$ is transitively closed, and by the assumption some member of $\calF$ is a subgraph of an $(H,I,D)$-generalized cactus.
It remains to find such a cactus inside $G$. The tester does not know $\calK$, so we separate the analysis into two parts. We first analyze an \emph{ideal process}, which picks one copy from $\calK$ for each step of the construction, and show that it finds a member of $\calF$  with good probability. We then show that the real tester can simulate this process (in constantly many steps).

\paragraph{The ``ideal" process.} In the setting of Lemmas~\ref{lem:collection} and~\ref{lem:reg}, fix the construction chosen above for $(H,I,D)$. It has steps $1,\ldots,N'$ with $N'\le N$, and its graph $F$ contains a copy of some $H^*\in\calF$. Step $m$ creates a petal $P^{(m)}$ isomorphic to the $I$-component $C_{\ell_m}(I)$, with the list $\lambda_m$. For $m\ge2$, step $m$ attaches $P^{(m)}$ to $P^{(j_m)}$ at the vertex $u_m:=\lambda_{j_m}(v_{i_m})$, for some $j_m<m$ and some role $v_{i_m}$. The \emph{ideal process} chooses random copies $\bK_1,\ldots,\bK_{N'}\in\calK$, one for each step, as follows.
\begin{itemize}
\item $\bK_1$ is uniform on $\calK$.
\item For $m\ge2$, the copy $\bK_m$ is uniform on the copies $K\in\calK$ whose hub of role $v_{i_m}$ equals that of $\bK_{j_m}$, that is, $K(v_{i_m})=\bK_{j_m}(v_{i_m})$, where $K(v_i)=\psi_K(v_i)$ as before.
\end{itemize}

For every entry $f$ of one of the lists $\lambda_1,\ldots,\lambda_{N'}$ of the fixed construction, we set $\bPsi(f):=\bK_m(v_\ell)$, where $m$ and $\ell$ are any indices with $f=\lambda_m(v_\ell)$. Here $f$ is the vertex playing the role $v_\ell$ in the $H$-copy of step $m$. It is a vertex of $F$ if some petal contains it, and otherwise it is only an entry of the lists. In both cases $\bPsi(f)$ is a hub of $G$. Intuitively, $\bPsi$ maps each vertex of the cactus that plays a role in $I$ to the hub of $G$ that plays this role in the corresponding chosen copy. A given entry $f$ may appear in several lists, so this definition involves a choice of $m$ and $\ell$. The role $v_\ell$ is always the same, since an entry is created for a single role and is copied only to the same role in later lists. The value $\bK_m(v_\ell)$ does not depend on the choice of $m$ either, by the next lemma. So $\bPsi(f)$ is well defined.

Let $\bP_m:=\psi_{\bK_m}(C_{\ell_m}(I))$ be the petal of $\bK_m$ corresponding to $C_{\ell_m}(I)$, and let $\bL_m:=\psi_{\bK_m}(V(C_{\ell_m}))$ be the set of its light vertices. The petals $\bP_m$ are glued in $G$ exactly as the petals $P^{(m)}$ are glued in $F$, unless they meet in more vertices than the construction prescribes. The following three events exclude this.
\begin{itemize}
\item $\calbE_1$: no copy $\bK_m$, $m\le N'$, contains a bad vertex (Definition~\ref{def:good}).
\item $\calbE_2$: for every $2\le m\le N'$ and every new entry $\lambda_m(v_\ell)$ of step $m$, the hub $\bK_m(v_\ell)$ is not a hub of any $\bK_{m'}$ with $m'<m$.
\item $\calbE_3$: for all $m'<m\le N'$, the set $\bL_m$ is disjoint from $V(\bP_{m'})$.
\end{itemize}
Let $\calbE:=\calbE_1\cap\calbE_2\cap\calbE_3$. For a fixed $m$, we say that one of these events \emph{fails at step $m$} if its condition is violated for this $m$.

\begin{lemma}\label{lem:psi_consistent}
If $\lambda_j(v_\ell)=\lambda_m(v_\ell)$ for two steps $j,m$, then $\bK_j(v_\ell)=\bK_m(v_\ell)$.
\end{lemma}
 
\begin{proof}
We use induction on $\max(j,m)$. Let $j<m$. The entry $\lambda_m(v_\ell)$ equals an earlier entry, so it is not new. By the definition of step $m$, we have $\lambda_m(v_\ell)=\lambda_{j_m}(v_\ell)$, and $v_\ell=v_{i_m}$ or $(v_{i_m},v_\ell)\in E_D$. By the choice of $\bK_m$, the copies $\bK_m$ and $\bK_{j_m}$ have the same hub of role $v_{i_m}$. If $v_\ell\ne v_{i_m}$, then $(v_{i_m},v_\ell)$ is locked in $\calK$ by Lemma~\ref{lem:reg}(a), so they also have the same hub of role $v_\ell$. In either case $\bK_m(v_\ell)=\bK_{j_m}(v_\ell)$. Since $\lambda_{j_m}(v_\ell)=\lambda_j(v_\ell)$ with $j_m,j<m$, induction gives $\bK_{j_m}(v_\ell)=\bK_j(v_\ell)$.
\end{proof}

\begin{lemma}\label{lem:uniform}
Each $\bK_m$ is uniformly distributed on $\calK$.
\end{lemma}
 
\begin{proof}
For $m=1$ the statement follows trivially by definition. Let $m\ge2$, so that $\bK_{j_m}$ is uniform by induction. For a vertex $g$ let $\calK_g:=\{K\in\calK:K(v_{i_m})=g\}$. Given $\bK_{j_m}(v_{i_m})=g$, the copy $\bK_m$ is uniform on $\calK_g$. So for a fixed $K\in\calK$ with $K(v_{i_m})=g$ we get
\[ \Pr[\bK_m=K]=\Pr[\bK_{j_m}\in\calK_g]\cdot\frac1{|\calK_g|}=\frac{|\calK_g|}{|\calK|}\cdot\frac1{|\calK_g|}=\frac1{|\calK|}.\qedhere \]
\end{proof}

\begin{lemma}\label{lem:embedding}
On $\calbE$, the graph $\bigcup_m\bP_m$ contains a copy of $H^*$.
\end{lemma}
 
\begin{proof}
Define $\Gamma:V(F)\to V(G)$ as follows. A vertex of the petal $P^{(m)}$ that is not an entry, corresponding to $x\in V(C_{\ell_m})$, is sent to $\psi_{\bK_m}(x)$. An entry $f$ that is a vertex of $F$ is sent to $\bPsi(f)$. Since $\bPsi(\lambda_m(v_\ell))=\bK_m(v_\ell)$, the petal $P^{(m)}$ is mapped isomorphically onto $\bP_m$. So $\Gamma$ is a homomorphism into $\bigcup_m\bP_m$, and it remains to show that $\Gamma$ is injective.
\begin{itemize}
\item Vertices that are not entries are mapped to light vertices, and entries to heavy vertices.
\item Vertices that are not entries and belong to the same petal are mapped injectively, and those of different petals are mapped to disjoint sets by $\calbE_3$.
\item Let $y\ne y'$ be entries with $\bPsi(y)=\bPsi(y')$. Since every heavy vertex has a single role, $y$ and $y'$ have the same role $v_\ell$. Let $i$ and $i'$ be the steps in which $y$ and $y'$ appear as new entries, so that $y=\lambda_i(v_\ell)$ and $y'=\lambda_{i'}(v_\ell)$, and say $i\le i'$. Then $i<i'$, since a step has only one entry of each role, and so $i'\ge2$. Now $\bPsi(y')=\bK_{i'}(v_\ell)$ equals $\bPsi(y)=\bK_i(v_\ell)$, which is a hub of $\bK_i$. This contradicts $\calbE_2$.
\end{itemize}
Hence $\Gamma(F)$ is isomorphic to $F\supseteq H^*$.
\end{proof}

\begin{lemma}\label{lem:ideal_prob}
Let $\theta\in(0,1)$, and suppose that $kN\tau\le\theta/N$, that $\eta\le\theta/N$, and that $h_{j+1}\ge Rh_j$ for all $j\in[2k]$, where $R\ge kN^2/(\beta\theta)$. Then $\Pr[\overline{\calbE}]\le3\theta$.
\end{lemma}

\begin{proof}
The event $\calbE_1$ fails at step $m$ exactly when $\bK_m$ contains a bad vertex. By Lemma~\ref{lem:uniform}, $\bK_m$ is uniform on $\calK$, and by Lemma~\ref{lem:reg}(b), at most $\eta|\calK|$ copies of $\calK$ contain a bad vertex. So $\calbE_1$ fails at step $m$ with probability at most $\eta\le\theta/N$.
 
Now fix $m\ge2$, and condition on the past, assuming that $\calbE_1$ does not fail at any earlier step. Let $u:=\bPsi(u_m)$.
It is a hub of $\bK_{j_m}$, so it is good: it lies in at least $\beta\deg(u)$ copies of $\calK$, and every pair $(v_{i_m},v_\ell)$ with
$\ell\ne i_m$ and $(v_{i_m},v_\ell)\notin E_D$ is $\tau$-spread at $u$. Moreover, $\bK_m$ is uniform among the $\deg_\calK(u)$ copies with
hub $u$ of role $v_{i_m}$.
\begin{itemize}
\item $\calbE_2$: a new entry of step $m$ has a role $v_\ell$ with $\ell\ne i_m$ and $(v_{i_m},v_\ell)\notin E_D$. For each of the at most $kN$ hubs $w$ of earlier copies, the probability that $\bK_m(v_\ell)=w$ is at most $c_\calK(u,w)/\deg_\calK(u)\le\tau$, since $u$ is good. So $\calbE_2$ fails at step $m$ with probability at most $kN\tau\le\theta/N$.
\item $\calbE_3$: for each of the at most $kN$ light vertices $w$ of earlier copies, the probability that $\bK_m$ contains $w$ is at most $h_b/\deg_\calK(u)\le h_b/(\beta h_{b+1})\le1/(\beta R)$, by Lemma~\ref{lem:collection}(d) and since $u$ is good and heavy. So $\calbE_3$ fails at step $m$ with probability at most $kN/(\beta R)\le\theta/N$.
\end{itemize}
Since $\calbE_2$ and $\calbE_3$ cannot fail at step $1$, a union bound over the at most $N$ steps and the three events gives $\Pr[\overline{\calbE}]\le3\theta$.
\end{proof}

Hence, the ideal process finds $H^*$ with probability at least $1-3\theta$. The tester cannot pick a uniform copy of $\calK$ (since it does not know $\calK$). It can, however, sample random vertices and random neighbors, and we now show that
these two operations simulate the two kinds of steps of the ideal process. A random vertex simulates step $1$, and a random neighbor of an expanded hub simulates an attached step. Intuitively this should be obvious: at each step of a real test, that has constructed so far the correct partial $m$-size cactus, the next correct petal can be discovered with constant probability. The formal details follow.

\paragraph{Realizing the ideal process.} A \emph{round with parameters $Z$, $L$ and $s$} samples $Z$ vertices uniformly at
random and runs \textsf{Bounded-BFS}$(G,\cdot,L,s)$ from each of them, where, as usual, vertices at a level below $L$ are
\emph{expanded}, that is, $s$ uniform neighbors of each of them are sampled. In the analysis, $\bK_1$ is a copy containing a
sampled vertex, and each later $\bK_m$ is a copy containing a sampled neighbor of the hub at which step $m$ is attached. With
probability at least $1-2\theta$, every expanded light vertex sees all of its neighbors (Lemma~\ref{lem:BFS}), some sampled
vertex is accepted in step $1$ (by the choice of $Z$), and every expanded good hub has a separate suitable sampled neighbor
for each step attached at it (by a Chernoff bound and the choice of $s$). We assume this throughout. A step is then
carried out unless its hub is bad, and step $m$ is \emph{realized} if steps $1,\ldots,m$ are carried out.
 
\begin{lemma}\label{lem:realization}
Let $\theta\in(0,1)$, and let $L$, $Z$ and $s$ be integers with
\begin{gather*}
L=N(k+2),\;\;\; Z\ge\frac{h_{2k+1}}{c_1}\ln(4/\theta),\;\;\;
s\ge\frac{4N}\beta\ln\left(\frac{4kN}\theta\right),\;\;\; h_{2k+1}(s+1)^Le^{-s/h_{2k+1}}\le\frac\theta Z.
\end{gather*}
Then the copies chosen from a round with parameters $Z$, $L$ and $s$ have the following properties.
\begin{enumerate}
\item[(a)] For every $m\le N'$ and every set $\calA$ of sequences of $m$ copies,
\[
\Pr\big[\text{step $m$ is realized and }(\bK_1,\ldots,\bK_m)\in\calA\big]\le\Pr\nolimits_{\mathrm{ideal}}\big[(\bK_1,\ldots,\bK_m)\in\calA\big],
\]
where $\Pr_{\mathrm{ideal}}$ refers to the ideal process.
\item[(b)] If $\calbE_1$ holds for the steps before $m$, then step $m$ is realized, and the round discovers all edges of $\bP_m$.
\end{enumerate}
\end{lemma}

\begin{proof}
The analysis follows the construction step by step inside one run of the round: step $1$ is realized by a sampled
vertex, and each attached step by a sampled neighbor of its attachment hub. Once a run reaches a vertex of $\bL_m$ at level
$\ell<L-k$, it reaches and expands the whole connected set $\bL_m$, which has fewer than $k$ vertices, all of them light, and
so it discovers all of $\bP_m$, within level $\ell+k$.

\emph{Step $1$.} For a vertex $v$, let $c_v$ be the number of copies $K\in\calK$ with $v\in\psi_K(V(C_{\ell_1}))$, so that $c_v<h_b$ by
Lemma~\ref{lem:collection}(d). In the analysis, we accept each $\bv_i$ independently with probability $c_{\bv_i}/h_b$, and let $\bK_1$ be
a uniform copy among the $c_v$ copies of the first accepted vertex $v$. A vertex in many copies is more likely to lead to a
given copy, and the acceptance probability compensates for this: for each $i$ and each $K\in\calK$,
\[
\Pr[\text{$\bv_i$ is accepted and $K$ is chosen}]=\sum_{v\in\psi_K(V(C_{\ell_1}))}\frac1n\cdot\frac{c_v}{h_b}\cdot\frac1{c_v}=\frac{|V(C_{\ell_1})|}{nh_b},
\]
which does not depend on $K$. So $\bK_1$ is uniform on $\calK$. Moreover, each $\bv_i$ is accepted with probability at least
$c_1/h_{2k+1}$, which is how $Z$ was chosen.

\emph{Steps $m\ge2$.} Let $u:=\bPsi(u_m)$. This is the hub of role $v_{i_m}$ of $\bK_{j_m}$, and, since $u_m$ is a vertex of $F$, it is
a vertex of some earlier petal $\bP_i$, $i<m$. Call a sampled neighbor of $u$ \emph{useful} for step $m$ if it is a light vertex
of the petal corresponding to $C_{\ell_m}(I)$ of some copy $K\in\calK$ with $K(v_{i_m})=u$. Every such copy contributes exactly
$\deg_{C_{\ell_m}(I)}(v_{i_m})\ge1$ such neighbors of $u$, and these are distinct for distinct copies, since the copies are
edge-disjoint. So the copy of a useful sample is uniform among the copies $K\in\calK$ with $K(v_{i_m})=u$, and each step attached
at $u$ takes the copy of the next unused useful sample. Since the samples are independent, given the history and that
step $m$ is realized, $\bK_m$ has the distribution of the ideal process.

\emph{Proof of (a).} Given any history in which steps $1,\ldots,m-1$ are realized, step $m$ is realized with some probability,
and then $\bK_m$ has the conditional distribution of the ideal process. So the probability that steps $1,\ldots,m$ are
realized with given copies is at most their probability in the ideal process, by induction on $m$. Summing over $\calA$
gives (a).

\emph{Proof of (b).} We show by induction on $m$ that step $m$ is realized and that all of $\bP_m$ is discovered within
level $m(k+2)$. For $m=1$, the accepted vertex lies in $\bL_1$ at level $0$, so $\bP_1$ is discovered within level $k$. For
$m\ge2$, all vertices of earlier petals lie at level at most $(m-1)(k+2)<L$, so $u$ is expanded. Since $\calbE_1$ holds before step
$m$ and $u$ is a hub of $\bK_{j_m}$, the hub $u$ is good, so it has enough useful samples, and step $m$ is realized. The sample it
uses is a vertex of $\bL_m$ at level at most $(m-1)(k+2)+1<L-k$, so $\bP_m$ is discovered within level $m(k+2)$.
\end{proof}

We now put everything together. The final lemma fixes all parameters and combines the ideal process with its realization.
\begin{lemma}[general upper bound]\label{lem:general_UB}
Let $\calF$ be a finite family of connected graphs, each with at least one edge and at most $k$ vertices, such that
for every $H\in\calF$, every nonempty labeled independent set $I\subseteq V(H)$ and every transitively closed digraph $D$ on
$I$, some member of $\calF$ is a subgraph of an $(H,I,D)$-generalized cactus. Then for every $p\ge1$ and $\eps>0$ there is
a one-sided $\eps$-tester for $\calF$-freeness on $p$-degenerate graphs, whose number of queries is $O_{p,\eps,\calF}(1)$.
\end{lemma}

\begin{proof}Recall from the assumptions that $N$ is the largest number of steps among the constructions fixed for the triples $(H,I,D)$, where $H\in\calF$, $I\subseteq V(H)$ is a nonempty labeled independent set, and $D$ is a transitively closed digraph on $I$. Such a construction exists for every triple by the hypothesis of the lemma, and its number of steps depends only on the triple, not on the input graph. Moreover, there are at most $|\calF|\,2^k\,k!\,2^{k(k-1)}$ triples. So $N$ is a constant that depends only on $\calF$, and not on the input graph or on $n$. The input determines which triple is used, and thus the number $N'\le N$ of steps of the construction that is followed, but the parameters below depend on $N$ only.

We choose the parameters in the following order, each as required by the indicated lemma.
\begin{enumerate}
\item $\theta:=\frac1{20}$, the target failure probability.
\item $\tau:=\theta/(kN^2)$ and $\delta:=\theta/(2k^2N)$, so that $kN\tau=\theta/N$ and $k^2\delta=\theta/(2N)$ (Lemma~\ref{lem:ideal_prob}).
\item $c_0:=\eps/(4k^2|\calF|2^kk^k)$ (Lemma~\ref{lem:collection}) and $c_1:=c_0(\tau\delta)^{k^2}$ (Lemma~\ref{lem:reg}).
\item $\beta:=\theta c_1/(4pN)$, so that $2\beta p/c_1=\theta/(2N)$ and hence $\eta=\theta/N$ (Lemma~\ref{lem:ideal_prob}).
\item $R:=kN^2/(\beta\theta)$, and $h_j:=\lceil4p^2/\eps\rceil R^{j-1}$ for $j\in[2k+1]$, as required by
Lemmas~\ref{lem:collection} and~\ref{lem:ideal_prob}.
\item $L:=N(k+2)$, $Z:=\lceil(h_{2k+1}/c_1)\ln(4/\theta)\rceil$, and an integer $s$ as in Lemma~\ref{lem:realization}.
\end{enumerate}
All of them depend only on $p$, $\eps$ and $\calF$. The tester performs one round with parameters $Z$, $L$ and $s$, and rejects if and only if the explored subgraph
contains a member of $\calF$. It makes at most $Zs(s+1)^L=O_{p,\eps,\calF}(1)$ queries, and it has one-sided error, since it
rejects only after finding a member of $\calF$. 
Let $G$ be $p$-degenerate and $\eps$-far from $\calF$-freeness. Apply Lemmas~\ref{lem:collection} and~\ref{lem:reg}, and let
$N'\le N$ be the number of steps of the construction fixed for the resulting triple $(H,I,D)$. Consider the round of the tester, and the copies obtained from it as described before Lemma~\ref{lem:realization}. If all
steps $1,\ldots,N'$ are realized, the round discovers all edges of every $\bP_m$, and $\calbE$ holds, then the round discovers a
member of $\calF$ by Lemma~\ref{lem:embedding}. So the round fails to find a member of $\calF$ only if one of the following two
events occurs.
\begin{itemize}
\item[(i)] For some $m\le N'$, step $m$ is realized, and $\calbE$ fails for the first time at step $m$.
\item[(ii)] One of the assumptions stated before Lemma~\ref{lem:realization} fails.
\end{itemize}
Indeed, if neither occurs, then $\calbE$ holds at every realized step by (i), so $\calbE_1$ holds before every step, and by
Lemma~\ref{lem:realization}(b) every step is realized and every $\bP_m$ is discovered.

The event (ii) has probability at most $2\theta$, as stated before Lemma~\ref{lem:realization}. We bound the probability of (i)
by comparing the round with the ideal process. For $m\le N'$, let $\calA_m$ be the set of sequences $(K_1,\ldots,K_m)$ of copies
for which the conditions of $\calbE_1$, $\calbE_2$ and $\calbE_3$ hold at the steps $1,\ldots,m-1$, and one of them fails at step $m$. So
$(\bK_1,\ldots,\bK_m)\in\calA_m$ means that $\calbE$ fails for the first time at step $m$, and the event (i) is the union, over $m\le N'$, of
the events that step $m$ is realized and $(\bK_1,\ldots,\bK_m)\in\calA_m$. The sets $\calA_m$ are disjoint, and in the ideal process
their union is the event $\overline{\calbE}$. Hence, by Lemma~\ref{lem:realization}(a) and Lemma~\ref{lem:ideal_prob},
\begin{align*}
\Pr[\text{(i)}]&\le\sum_{m\le N'}\Pr\big[\text{step $m$ is realized and }(\bK_1,\ldots,\bK_m)\in\calA_m\big]\\
&\le\sum_{m\le N'}\Pr\nolimits_{\mathrm{ideal}}\big[(\bK_1,\ldots,\bK_m)\in\calA_m\big]=\Pr\nolimits_{\mathrm{ideal}}\big[\overline{\calbE}\big]\le3\theta.
\end{align*}
Here only an upper bound is needed: Lemma~\ref{lem:realization}(a) is used to bound the probability of the bad
configurations $\calA_m$, and a realized step is at most as likely as in the ideal process, since the round may stop at a
bad hub.

Hence the round fails to find a member of $\calF$ with probability at most $3\theta+2\theta=5\theta=1/4$, so the tester rejects
with probability at least $3/4\ge2/3$.
\end{proof}

Combining the upper bound with the general lower bound gives Theorem~\ref{thm:main}.
\begin{proof}[Proof of Theorem~\ref{thm:main}]
If the condition holds, then $\calF$-freeness is testable by Lemma~\ref{lem:general_UB}. If it fails, then there are $H\in\calF$,
a nonempty labeled independent set $I\subseteq V(H)$ and a transitively closed digraph $D$ on $I$ such that no member of $\calF$ is
a subgraph of any $(H,I,D)$-generalized cactus. Lemma~\ref{thm:LB_gen}, applied to this triple, shows that $\calF$-freeness is
not testable on $p$-degenerate graphs, since $p\ge k-1$.
\end{proof}

\section{Further Discussion}\label{sec:further}

Theorem~\ref{thm:main} which characterizes testable families is rather implicit. Here we give several outcomes that are implied by it.

\begin{enumerate}
    \item The following is a necessary condition for testability of $\calF = \{H_1,H\}$ where $H_1$ is $2$-connected and not testable (hence contain an obstacle). 

    \begin{definition}[Cactus with respect to $(H_1, S)$]\label{def:cactus}
Let $H_1$ be a $2$-connected graph and $S \subset V(H_1)$ an obstacle. Let $\mathcal{C}$ be the components of $H_1 \setminus S$. A \emph{cactus} with respect to $(H_1,S)$ is a 
 graph, with a set of articulation points $L$ (not necessarily maximal) such that: 
  Each $L$-component of $H$ is a isomorphic to a subgraph of an $S$-component $C \in \calC$ (the components $C \in \calC$ are not necessarily distinct.) 
\end{definition}

See e.g., Figure~\ref{fig:fig71}. The graph $H_1$ has an obstacle $S = \{a,b,c\}$ with two components $C_1,C_2$. The graph $H$ is an $(H_1,S)$-cactus, with $L=\{a,b,c\}$ and $L$-components $P_1,P_2,P_3$. 

\begin{lemma}
    \label{lem:71}
    Let $H_1$ be a $2$-connected graph. A necessary condition for a family $\calH = \{H_1,H\}$ to be testable, is that $H$ is an $(H_1,S)$-cactus with respect to every obstacle $S$.
\end{lemma}
The family $\calF = \{H_1,H\}$ from Figure~\ref{fig:fig71}  
meets the above necessary condition: there are two obstacles, one of size $2$ (the two neighbors of $c$, for which the condition is easy to check), and $S=\{a,b,c\}$, as depicted in the figure, and with the corresponding petals. 

The reason that the lemma is a necessary condition is that for $I=S\subset V(H_1)$, for an obstacle $S$, the lower bound presented in Section~\ref{sec:2-conn-main}, w.r.t to $H_1$ still forms a lower bound if $H$ is not a cactus w.r.t $S$, as it will not be found in the graph constructed for the above lower bound.

We note that the lemma also 
follows from Theorem~\ref{thm:main}  for $I=S$ and the empty dependency graph $D$.
\begin{figure}[ht!]
\centering
\begin{tikzpicture}[
  v/.style={circle,fill=black,inner sep=0pt,minimum size=5pt},
  lab/.style={font=\large}]
\begin{scope}
  \fill[gray!25,rounded corners=12pt] (-0.55,-0.55) -- (-0.6,2.55) -- (0.85,2.55) -- (1.0,2.2) -- (1.0,-0.2) -- (0.9,-0.55) -- cycle;
  \fill[gray!25,rounded corners=12pt] (1.1,2.55) -- (2.45,2.4) -- (2.55,1.0) -- (2.0,-0.5) -- (1.2,-0.5) -- (1.1,-0.1) -- (1.1,2.2) -- cycle;
  \node[v] (x1) at (0,2) {};
  \node[v] (x2) at (0,1) {};
  \node[v] (x3) at (0,0) {};
  \node[v,label={[lab]above:$a$}] (a) at (1.05,2) {};
  \node[v,label={[lab]above:$b$}] (b) at (1.05,1) {};
  \node[v,label={[lab]above:$c$}] (c) at (1.05,0) {};
  \node[v] (r) at (2.05,1) {};
  \draw[thick] (x1)--(a) (x1)--(x2) (x2)--(x3) (x1)--(b) (x2)--(a) (x2)--(b) (x3)--(b) (x3)--(c)
               (a)--(r) (b)--(r) (c)--(r);
  \node[lab] at (0.1,3.0) {$C_1$};
  \node[lab] at (1.9,3.0) {$C_2$};
  \node[font=\Large] at (1.0,-1.2) {$H_1$};
\end{scope}
\begin{scope}[xshift=5.2cm]
  \colorlet{pf}{blue!20}\colorlet{pd}{blue!60!black}
  \filldraw[fill=pf,draw=pd,rounded corners=10pt] (-0.3,1.45) rectangle (3.0,2.75);
  \filldraw[fill=pf,draw=pd] (0,1.42) to[out=-15,in=70] (0.2,0.35) to[out=-110,in=-70] (-0.2,0.35) to[out=110,in=-165] (0,1.42);
  \filldraw[fill=pf,draw=pd,rounded corners=10pt] (1.0,0.2) rectangle (2.45,1.45);
  \node[v] (t1) at (0.4,2.4) {};
  \node[v] (t2) at (1.7,2.4) {};
  \node[v] (t3) at (2.7,2.4) {};
  \node[v,label={[lab]left:$a$}] (a) at (0,1.55) {};
  \node[v,label={[lab,label distance=1pt]left:$b$}] (b) at (1.7,1.55) {};
  \node[v,label={[lab]below right:$c$}] (c) at (2.7,1.55) {};
  \node[v] (p) at (0,0.55) {};
  \node[v] (q1) at (1.25,0.55) {};
  \node[v] (q2) at (1.7,0.55) {};
  \node[v] (q3) at (2.15,0.55) {};
  \draw[thick] (t1)--(t2) (t2)--(t3) (a)--(t1) (a)--(t2) (t1)--(b) (t2)--(b) (t3)--(b) (t3)--(c)
               (a)--(p) (b)--(q1) (b)--(q2) (b)--(q3) (q1)--(q2) (q2)--(q3);
  \node[lab] at (1.5,3.05) {$P_2$};
  \node[lab] at (0,-0.1) {$P_1$};
  \node[lab] at (1.7,-0.25) {$P_3$};
  \node[font=\Large] at (1.4,-1.2) {$H$};
\end{scope}
\end{tikzpicture}
\caption{Cactus with respect to $(H_1,S)$.}\label{fig:fig71}
\end{figure}
\item The reason that Lemma~\ref{lem:71} is not a sufficient condition is that if we add a long enough chain of such $L$-components, then some of this components will need to connect via distinct  vertices taking the {\em same} $S$-role, say  $a$. This however, may not be possible to be discovered if in the dependency graph of the input graph $G$ has every $b,c$-role vertices point to  $a$. Namely, the same vertex $u \in G$ should appear in every petal in the cactus $H$, hence not allowing distinct vertices taking this role. Indeed for such a case, $H$ will not be a subgraph of the generalized cactus as required in Theorem~\ref{thm:main} (this however does not immediately implies a lower bound, as $H_1$ itself maybe easy to discover for such a construction). 

\item When $H_1$ is $2$-connected then Theorem~\ref{thm:2} characterize when $H_1$ is testable. But what about a graph $H$ that is composed by gluing together graphs $H_1,H_2$ at an articulation point $v$. 

It is obvious (and also follow from Theorem~\ref{thm:main}) that if both $H_1, H_2$ are testable, then so is $H$. But what if $H_1$ is $2$-connected and has an obstacle $S$, while $H_2$ is testable.

A necessary condition for $H$ to be testable in this case is that the articulation point $v$ of $H_1$ needs to be in {\em every} minimum size obstacle $S$ of $H_1$. While this is not a sufficient condition, it is easy to see its necessity: if $v$ is not in an obstacle $S$, then again, the lower bound given by the graph in the proof of Theorem~\ref{thm:2}, where for each copy of $H_1$, a disjoint copy of $H_2$ is glued on the vertex taking the $v$-role, 
will not find $H$, as it then would also finds $H_1$ in the lower bound graph, contradicting Theorem~\ref{thm:2}.

To just show an instance in which $H$-freeness is testable,  when $H$ is formed by disjoint $H_1, H_2$ glued at $v$, consider the following. $H_2$ is just $K_8$ which is 2-connected and testable. $H_1$ is composed of two copies of $K_4$ where in each one edge is removed, and they are glued together at this two, now non-adjacent vertices. 
$H_1$ is 2-connected non-testable as the two vertices on which the copies are glued, say $u,v$ form an obstacle. 

Now let $H$ be $H_1, H_2$ glued together at $v$ - one of the vertices adjacent to the removed edge in $H_1$.  Note that $H_1$ is a subgraph of $H_2$. Hence $H$ will be testable, as in every graph containing many $H$ copies, either the graph is easy to test for $H_1$ freeness, and then an adjacent $K_8$ would also be easily found, or, it is hard to test graph for $H_1$, but then many copies of $H$ containing $H_1$ must share the vertex of role $v$ above. 
Sampling $v$ may not lead to the discovery of $H_1$, but it will lead to the discovery of two $H_2=K_8$ sharing $v$ and hence containing $H$. \end{enumerate}
\paragraph{Acknowledgments}
We thank the anonymous SODA 2027 reviewers for their helpful comments. Claude Opus 5.5  was used to assist with verifying proofs, preparing figures, and styling the text. The authors take full responsibility for all results.
\appendix
\section{Bounded BFS}\label{sec:BFS}

To explore the neighborhood of a vertex we use the following subroutine, which
simulates a breadth-first search of bounded depth with the random neighbor oracle. At every vertex it
reaches, it queries $s$ random neighbors. If $s$ is large compared with the degree of a vertex, then with high
probability all neighbors of that vertex are found. For vertices of larger degree, only a sample of the
neighborhood is seen.

\begin{figure}[ht!]
\begin{framed}
\noindent Subroutine \textsf{Bounded-BFS}$(G,v,t,s)$.
\begin{flushleft}
\noindent {\bf Input:} random neighbor oracle access to $G=(V,E)$, a starting vertex $v\in V$, a depth bound
$t$, and the number $s$ of neighbor queries at each vertex.\\
{\bf Output:} a subgraph of $G$.
\begin{enumerate}
\item Let $\bS_0=\{v\}$.
\item For $\ell=1$ to $t$:
\begin{itemize}
\item Let $\bS_\ell=\emptyset$ and $\mathbf{E}_\ell=\emptyset$.
\item For every $u\in \bS_{\ell-1}$, choose neighbors $\bw_1,\ldots,\bw_s$ of $u$ independently and uniformly
at random, and add the edges $u\bw_i$ to $\mathbf{E}_\ell$ and the vertices $\bw_i$ to $\bS_\ell$.
\item Let $\bS_\ell:=\bS_\ell\setminus\bigcup_{i=0}^{\ell-1}\bS_i$.
\end{itemize}
\item {\bf Return:} the graph formed by the edges in $\bigcup_{\ell=1}^t\mathbf{E}_\ell$.
\end{enumerate}
\end{flushleft}\vskip -0.14in
\end{framed}\vspace{-0.2cm}
\caption{Simulation of a bounded-depth BFS using the random neighbor oracle.}\label{fig:simulation}
\end{figure}

A vertex is \emph{expanded} if it lies in $\bS_\ell$ for some $\ell<t$, that is, if its neighbors are queried.
Every vertex lies in at most one set $\bS_\ell$, so it is expanded at most once. Moreover, if all vertices on a
path from $v$ of length $\ell\le t$ are expanded and the edges of the path are among the sampled edges, then
the endpoint of the path lies in $S_{\ell'}$ for some $\ell'\le\ell$. The following lemma is standard.

\begin{lemma}\label{lem:BFS}
Let $h,t\ge1$ and $\delta\in(0,1)$, and let $s$ be an integer with $h(s+1)^te^{-s/h}\le\delta$ (In particular, one can take $s=O(ht\log(ht/\delta))$).  Then,
with probability at least $1-\delta$, every expanded vertex of degree at most $h$ has all of its neighbors
among its $s$ sampled neighbors in \textsf{Bounded-BFS}$(G,v,t,s)$. The subroutine makes at most
$s(s+1)^t$ queries.
\end{lemma}

\bibliography{biblio}
\bibliographystyle{alpha}
\newpage

\end{document}